\documentclass[preprint,showpacs,preprintnumbers,amsmath,amssymb]{revtex4}

\usepackage{graphicx}
\usepackage{dcolumn}
\usepackage{bm}
\usepackage{amsmath}
\usepackage{amssymb}
\usepackage{mathrsfs}
\usepackage{color}

\newtheorem{theorem}{{\bf Theorem}}
\newtheorem{proposition}{{\bf Proposition}}
\newtheorem{corollary}{{\bf Corollary}}

\newtheorem{definition}{{\bf Definition}}
\newtheorem{remark}{{\bf Remark}}

\def \rmd{\mathrm{d}}
\def \rme{\mathrm{e}}

\def \ad{\mathrm{ad}}

\def \ddt{\frac{\rmd}{\rmd t}}

\newenvironment{proof}[1][Proof]{\begin{trivlist}
\item[\hskip \labelsep {\bfseries #1}]}{\end{trivlist}}
\newcommand{\qed}{\nobreak \ifvmode \relax \else
      \ifdim\lastskip<1.5em \hskip-\lastskip
      \hskip1.5em plus0em minus0.5em \fi \nobreak
      \vrule height0.75em width0.5em depth0.25em\fi}

\newcommand{\dbracket}[2]{ {[}\mkern-2mu{[} #1 , #2 {]}\mkern-2mu{]} }
\newcommand{\dpoisson}[2]{ {\{}\mkern-5mu{\{} \, #1 , #2 \,{\}}\mkern-5mu{\}} }

\newcommand{\rn}[1]
    {\MakeUppercase{\romannumeral #1}}
 \newcommand{\RN}[1]
    {\mathrm{\MakeUppercase{\romannumeral #1}}}

\begin{document}
\title{Deformation of Lie-Poisson algebras and chirality }

\author{Zensho Yoshida$^1$ and Philip J. Morrison$^2$}
\affiliation{
$^1$ Department of Advanced Energy, University of Tokyo, Kashiwa, Chiba 277-8561, Japan
\\
$^2$ Department of Physics and Institute for Fusion Studies, University of Texas at Austin, TX 78712-1060, USA}

\email{yoshida@ppl.k.u-tokyo.ac.jp, morrison@physics.utexas.edu}

\begin{abstract}
Linearization of a Hamiltonian system around an equilibrium point yields a set of Hamiltonian-symmetric spectra:
If $\lambda$ is an eigenvalue of the linearized generator,
$-\lambda$ and $\overline{\lambda}$ (hence, $-\overline{\lambda}$) are also eigenvalues
--- the former implies a time-reversal symmetry, while the latter guarantees the reality of the solution. 
However, linearization around a \emph{singular equilibrium point} 
(which commonly exists in noncanonical Hamiltonian systems) works out differently,
resulting in breaking of the Hamiltonian symmetry of spectra;
time-reversal asymmetry causes \emph{chirality}.
This interesting phenomenon was first found in analyzing the chiral motion of the rattleback, a boat-shaped top having misaligned axes of inertia and geometry [Phys. Lett. A \textbf{381} (2017),  2772--2777].
To elucidate how chiral spectra are generated, we study the 3-dimensional Lie-Poisson systems,
and classify the prototypes of singularities that cause symmetry breaking.
The central idea is the \emph{deformation} of the underlying Lie algebra;
invoking Bianchi's list of all 3-dimensional Lie algebras,
we show that the so-called class-B algebras, 
which are produced by asymmetric deformations of the simple algebra $\mathfrak{so}(3)$,
yield chiral spectra when linearized around their singularities.
The theory of deformation is generalized to higher dimensions, 
including the infinite-dimensional Poisson manifolds relevant to fluid mechanics.
\end{abstract}


\pacs{02.20.Sv,02.20.Tw,02.40.Yy,45.20.Jj,47.10.Df}

\date{\today}
\maketitle


\section{Introduction}
\label{sec:intro}

Canonical Hamiltonian mechanics serves an archetype for  physics theories with two ingredients: 
 the energy (Hamiltonian) characterizing the mechanical property of matter 
and  symplectic geometry dictating the universal rules of kinematics in the phase space.
When  encountering peculiar dynamics, we might attribute it to a weird Hamiltonian, but usually 
we do not ascribe it to an  adjusted geometry of phase space.   Although this is the natural approach for understanding microscopic (canonical) mechanics,
the other perspective, i.e. deforming the geometry of phase space (keeping the Hamiltonian simple),  can be more effective for studying macroscopic systems in which 
some topological constraints foliate the phase space (such systems are called \emph{noncanonical}).
For example, a holonomic constraint reduces the effective phase space to a leaf embedded in the original canonical phase space, on which some interesting Lie algebra may dictate the kinematics. There are also many examples of noncanonical Hamiltonian systems in fluid and plasma physics,
where the essence of mechanics is attributed to complex Poisson brackets, while the Hamiltonians are rather simple\,\cite{Morrison1998}.

Here we explore the possibility of explaining  peculiar phenomena by the \emph{deformation} of phase space geometry.
In the next section, we start by reviewing a  model of the rattleback,
a boat-shaped top having misaligned axes of inertia and geometry\,\cite{Moffatt-Tokieda2008},
which is a 3-dimensional noncanonical Hamiltonian system
endowed with an interesting Poisson bracket\,\cite{YTM-rattleback}.
We find that the foliation of the 3-dimensional phase space by the Casimir invariant of this system has a \emph{singularity},
which turns out to be the cause of the symmetry breaking; viz.\ 
the linearized equations  obtained by expanding  about the singularity  have a pair of unbalanced positive and negative eigenvalues,
which explains the chirality (time-reversal asymmetry) of the rattleback.  
Here we also review  the basic formalism of Lie-Poisson manifolds.
Particular attention will be drawn to the duality of the space of state vectors (tangent bundle) and the phase space of observables (cotangent bundle),
which plays an essential role when we discuss the deformation of Lie algebras\,\cite{deformation}
and its reflection on Lie-Poisson brackets.
Using Bianchi's list of 3-dimensional Lie algebras (for example, see\,\cite{Ellis,Ryan}), 
we  examine all  3-dimensional Lie-Poisson manifolds;  (Sec.\,\ref{section:3D}).
Then, we  find an interesting correspondence between the classification of Lie algebras and symmetry breaking; 
class-A Lie-Poisson systems maintain the spectral symmetry, while class-B systems do not.
In Sec.\,\ref{sec:deformation},  we nail down the underlying structure that causes the symmetry breaking (chirality).
We show that all 3-dimensional Lie algebras are derived by deformation of a simple algebra $\mathfrak{so}(3)$,
and the asymmetry of the deformation endomorphism  brings about the symmetry breaking.
The 3-dimensional Lie algebras are special in that all possible Lie algebras are derived from a mother class-A  simple algebra by deformations; a symmetric endomorphism produces a class-A algebra, while an asymmetric deformation yields a class-B algebra.
In higher dimensions, we find  it necessary to  add another group of Lie algebras (called class C) that are not produced from some mother class-A algebra; we  find that their corresponding Lie-Poisson systems have chirality (Sec.\,\ref{sec:dimension>3}).

Because of  the richness of noncanonical Hamiltonian systems of fluids and plasmas,
we will construct, in Sec.\,\ref{sec:so(3)_bundle}, 
a bridge between the foregoing 3-dimensional Lie algebras and continuum (infinite-dimensional) dynamical systems.
Upon introducing a base space, we  define a vector bundle of Lie-algebra fibers to formulate field theories;   these theories 
are not yet relevant to fluid/plasma mechanics.  However, we  do  show that the deformation of the $\mathfrak{so}(3)$ bundle by the curl operator (which can be regarded as a symmetric deformation) yields the Lie-Poisson bracket of vortex dynamics.

\section{Preliminaries}
\label{sec:preliminaries}

\subsection{An example of chiral dynamics:  the rattleback}
\label{sec:PRS}

The peculiar motion of a rattleback
is  ``strange''  when examined in the light of Hamiltonian mechanics.
Let us start with a short review of our previous discussion\,\cite{YTM-rattleback}.

In the limit of zero dissipation, the equations of Moffatt and Tokieda \,\cite{Moffatt-Tokieda2008} are 
\begin{equation}
\frac{\rmd}{\rmd t} \left( \begin{array}{c}
P \\ R \\ S
\end{array} \right)
=
\left( \begin{array}{c}
\alpha P S \\ - RS \\ R^2 - \alpha P^2
\end{array} \right) ,
\label{MT-1}
\end{equation}
where $P, R, S$ stand for   {\it pitching}, {\it rolling}, and {\it spinning\/} modes of  motion.  
We will call (\ref{MT-1}) the PRS system and denote   the state vector by $\bm{x} = (P~R~S)^{\rm T} \in \mathbb{R}^3$.
The parameter $\alpha$ encodes the aspect ratio of the rattleback shape.
Here we assume $\alpha>1$, so that $P$ corresponds to lengthwise oscillations along the keel 
of the boat and  $R$ to sideways oscillations. 

Evidently, a purely spinning state $\bm{x}_s = (0~0~S_e)^{\rm T}$ 
($S_e$  an arbitrary constant measure of spin) is an equilibrium (steady state) of (\ref{MT-1}).
For small perturbation $\tilde{\bm{x}}=(\tilde{P}~\tilde{R}~\tilde{S})^{\rm T} $ around $\bm{x}_s$,  linearization of  (\ref{MT-1}) gives
\begin{equation}
\frac{\rmd}{\rmd t} 
\left( \begin{array}{c}
\tilde{P} \\ \tilde{R} \\ \tilde{S}
\end{array} \right)
=
\left(
\begin{array}{ccc}
\alpha S_e & 0 & 0
\\
0 & - S_e & 0
\\
0 & 0 & 0
\end{array} \right)
\left( \begin{array}{c}
\tilde{P} \\ \tilde{R} \\ \tilde{S}
\end{array} \right) .
\label{SeqEom}
\end{equation}
Hence, we obtain an  unbalanced spectrum with the eigenvalues  (time constants) 
$\alpha S_e$ and $-S_e$.
For $S_e > 0$, $\tilde{P}$ grows exponentially at the larger rate $\alpha S_e$,  while 
$\tilde{R}$ decays exponentially at the  smaller rate $S_e$.  
For $S_e < 0$, these are reversed, i.e.,
$\tilde{R}$ grows at the smaller rate $|S_e|$,  while $\tilde{P}$ decays at the larger rate $\alpha |S_e|$.
The chirality of rattleback motion manifests as
these unbalanced eigenvalue\,\cite{Moffatt-Tokieda2008}.

This observation raises a paradox, if we notice that the PRS system is a Hamiltonian system.
As is well-known, the spectra of a linearized Hamiltonian system must have   \emph{Hamiltonian symmetry},
i.e., when $\lambda$ is an eigenvalue, $-\lambda$, as well as the complex conjugate $\overline{\lambda}$ (hence, $-\overline{\lambda}$ also) 
are simultaneous eigenvalues.
The pair of $\lambda$ and $-\lambda$ guarantees a  time-reversal symmetry,
and the pair of $\lambda$ and $\overline{\lambda}$ guarantees the reality of the state vectors.
The unbalanced spectrum of (\ref{SeqEom}) does not have this symmetry.
However, we do find that (\ref{MT-1}) can be put into a Hamiltonian form\,\cite{YTM-rattleback}:
\begin{equation}
\frac{\rmd}{\rmd t} \bm{x} = J \partial_{\bm{x}} H,
\label{Hamiltonian_eq_PRS}
\end{equation}
where 
\begin{equation}
H(\bm{x})
=  \frac{1}{2}\left( P^2 + R^2 + S^2 \right)
\label{PRS_Hamiltonian}
\end{equation}
is the Hamiltonian, and 
\begin{equation}
J  = \left( \begin{array}{ccc}
0  &  0  &  \alpha P
\\
0  &  0  &  -R
\\
-\alpha P & R  & 0
\end{array} \right)
\label{MT-J}
\end{equation}
is the Poisson matrix (also called co-symplectic matrix and  Hamiltonian bi-vector).
By direct calculation, we can verify that Jacobi's identity holds for the bracket  
\begin{equation}
\{G,  H \} = ( \partial_{\bm{x}}G, J \partial_{\bm{x}} H ) ,
\label{MT-Poisson}
\end{equation}
where $(\bm{a},\bm{b})$ denotes the standard  inner-product (pairing).

\subsection{Lie-Poisson brackets}
\label{subsec:Lie-Poisson}

While we first derived the Poisson bracket (\ref{MT-Poisson}) through an heuristic argument,
there is a systematic method for constructing Poisson brackets from any given Lie algebra.
Such brackets are called \emph{Lie-Poisson brackets}, because they were known to Lie in the 19th century.
The PRS bracket (\ref{MT-Poisson}) was identified   as the type-VI Lie-Poisson bracket
in accordance with Bianchi's classification of 3-dimensional Lie algebras\,\cite{YTM-rattleback}.

\subsubsection{Phase space and measurement:}

Here we review the formulation of Lie-Poisson brackets,  paying attention to the relationship between the space of state vectors 
and its dual, i.e.\  the set of phase space of observables.
Let $X$ be a real vector space, which we call the  \emph{state space}.
When the dimension of $X$ is infinite, we assume that $X$ is a Banach space endowed with a norm $\| \cdot \|$.
A member $\bm{x}$ of $X$ is called a \emph{state vector}.
The dual space of $X$ (the vector space of linear functionals on $X$) is denoted by $X^*$.
With a bilinear pairing $\langle ~,~ \rangle:\, X\times X^* \rightarrow \mathbb{R}$, 
we can represent a linear functional as $\Xi(\bm{x}) = \langle \bm{x}, \bm{\xi} \rangle $ ($\bm{\xi}\in X^*$).
For a Hilbert space, we have the Riesz representation theorem,
so that we can identify $X^*=X$ by using the inner product $(~,~)$  in place of $\langle~,~\rangle$.
Physically, $\bm{\xi}$ means an \emph{observable}
 ($\Xi(\bm{x})$ is the measurement of a physical quantity for a state $\bm{x}$).
 We call $X^*$ the \emph{phase space}.

By introducing a basis for each space,
let us examine the mutual relationship between $X$ and $X^*$ more explicitly.
When $X$ has a finite dimension $n$, we can define a basis $\{\bm{e}_1,\cdots,\bm{e}_n \}$
to represent $\bm{x} =  x^k \bm{e}_k$
(we invoke Einstein's summation rule of contraction).
On the other hand,  we provide $X^*$ with the dual basis $\{\bm{e}^1,\cdots,\bm{e}^n \}$ such that
$\langle \bm{e}_j, \bm{e}^k \rangle = \delta_{jk}$.
A complete system of measurements is given by $\Xi^k (\bm{x})= \langle \bm{x} , \bm{e}^k \rangle $
($k=1,\cdots,n$);
measuring every $x^k = \langle \bm{x} , \bm{e}^k \rangle$ for a state vector $\bm{x}$, 
we can identify it as $\bm{x}= x^k \bm{e}_k$.
Therefore, we may say that $X^*$ defines $X$ as $(X^*)^*$
(this \emph{reflexive relation} is not trivial in infinite dimensions).
It is more legitimate to construct a theory by first defining $X^*$, 
since the description of a system depends on what we can measure.
For example, if we remove $\bm{e}^n$ from $X^*$, the component $x^n$ becomes invisible,
resulting in a reduced identification of $\bm{x}$ by only $x^j=\langle \bm{x}, \bm{e}^j\rangle$ ($j=1,\cdots, n-1$).


\subsubsection{The Lie algebra $X$:}

Endowing $X$ with a Lie bracket $[~ , ~ ]: X \times X \rightarrow X$, makes  $X$  a Lie-algebra.
By $\ad_{\bm{v}} \circ = [ \circ, \bm{v}] : \,X \rightarrow X$, we denote the adjoint representation of $\bm{v}\in X$.
Physically, the action of $\ad_{\bm{v}}$ on a state vector $\bm{x}\in X$ 
is a representation of infinitesimal dynamics:
\[
\ddt \bm{x}= \ad_{\bm{v}} \bm{x} = [\bm{x}, \bm{v}] .
\]

Dual to  $\ad_{\bm{v}}$,
we define the coadjoint action $\ad^*_{\bm{v}} \circ = [\bm{v},\circ]^*:\, X^* \rightarrow X^*$,
where $[~,~]^*: X \times X^* \rightarrow X^*$ is defined by 
\begin{equation}
\langle \bm{x}, [\bm{v}, \bm{\xi}]^* \rangle := \langle [\bm{x},\bm{v}], \bm{\xi} \rangle.
\label{coadjoint}
\end{equation}
The right-hand side means that we observe the dynamics of a state vector $\bm{x}$
by measuring an observable $\bm{\xi}$.
The left-hand side is its translation into the change in the observable $\bm{\xi}$ (evaluated for a fixed state vector $\bm{x}$):
\[
\ddt {\bm{\xi}}=\ad^*_{\bm{v}} \bm{\xi} = [\bm{v},\bm{\xi}]^*.
\]

\begin{remark}[semi-simple Lie algebra]
\label{remark:semi-simple}
\normalfont
For a semi-simple Lie algebra, we can formally evaluate $[\bm{x},\bm{y}]^* = [\bm{x},\bm{y}]$.
This means that we may identify $X=X^*$ (with an appropriate basis of $X^*$ as explained below),
and that the structure constants are \emph{fully antisymmetric}.
Let us write 
\begin{equation}
[\bm{e}_j,\bm{e}_k] = c_{jk}^\alpha \bm{e}_\alpha
\label{structure_constant}
\end{equation}
to define the structure constants $c_{jk}^\alpha$
($\{ \bm{e}_j \}$ is the basis of $X$).  
For a semi-simple Lie algebra, the Killing form 
$g_{jk} = c_{ja}^b c_{kb}^a$ is regular (nondegenerate).
We find that $c_{ijk} := c^\alpha_{jk} g_{\alpha i}$ is fully anti-symmetric
(i.e. $c_{jik}=c_{ikj}=-c_{ijk}$), by which the brackets $[~,~]$ and $[~,~]^*$ are equally evaluated as to be shown below in (\ref{fully-antisymmetric}).
Indeed, we observe, using Jacobi's identity,
\[
c_{ijk} = c_{jk}^\alpha c_{b\alpha}^a c_{ai}^b
= - (c_{kb}^\alpha c_{j\alpha}^a + c_{bj}^\alpha c_{k\alpha}^a) c_{ai}^b.
\]
Changing the indexes in the second term as $\alpha\mapsto b$, $a\mapsto \alpha$ and $b\mapsto a$, 
we may rewrite the right-had side as
\[
-c_{kb}^\alpha (c_{ai}^b c_{j\alpha}^a + c_{aj}^b c_{\alpha i}^a)
=
c_{kb}^\alpha c_{a\alpha}^b c_{ij}^a 
= g_{ak} c_{ij}^a = c_{kij}.
\]
We represent $\bm{\xi}\in X^*$ in contravariant variables as 
$\bm{\xi}=\xi^i \bm{\epsilon}_i :=\xi^i  g_{i \beta}\bm{e}^\beta$
($\bm{e}^\beta$ being the dual of $\bm{e}_\beta$).
Notice that this transformation is possible only when $g_{i\beta}$ is nondegenerate, i.e.\ $X$ is semi-simple.
We may write $\langle \bm{x}, \bm{\xi} \rangle = x^j \xi^i g_{ji}$.
For example, when $c_{jk}^i = \varepsilon_{ijk}$, the (scaled) Killing form  is $g_{ji}=-\delta_{ji}$.
Hence, $\xi^i = - \xi_i$, and $\langle \bm{x}, \bm{\xi} \rangle = -x^j \xi^i \delta_{ji} = x^j \xi_j $.
With contravariant variables $\bm{x}=x^j \bm{e}_j$, $\bm{y}=y^k \bm{e}_k$ 
and $\bm{\xi}= \xi^i \bm{\epsilon}_i =\xi^i  g_{i \beta}\bm{e}^\beta$, we may calculate
\[
\langle [\bm{x},\bm{y}], \bm{\xi}\rangle =
\langle c_{jk}^\alpha x^j y^k  \bm{e}_\alpha, \xi^i g_{i \beta } \bm{e}^\beta\rangle
=  c_{jk}^\alpha x^j y^k \xi^i g_{i \beta } \delta_{\alpha\beta} = c_{ijk} x^j y^k \xi^i.
\]
To write this as $\langle \bm{x}, [\bm{y}, \bm{\xi}]^* \rangle$,
we identify $[~,~]^*: X\times X^*\rightarrow X^*$ as (in the contravariant parameterization of $X^*$
with the basis $\{ \bm{\epsilon}_\alpha =  g_{\alpha j}  \bm{e}^j \}$)
\begin{equation}
[\bm{y}, \bm{\xi}]^* = c_{jki} y^k \xi^i \bm{e}^j = 
c_{ki}^\alpha y^k \xi^i g_{\alpha j}  \bm{e}^j
= c_{ki}^\alpha y^k \xi^i \bm{\epsilon}_\alpha 
= [\bm{y}, \bm{\xi}].
\label{fully-antisymmetric}
\end{equation}
We note that the general dual bracket $[~,~]^*$ is not necessarily a Lie bracket 
(even $[\bm{\xi},\bm{\xi}]^* =0$ may not hold).
\end{remark}

\subsubsection{The Poisson manifold $X^*$ and Lie-Poisson algebra:}

Now we construct a Poisson algebra on $C^\infty(X^*)$,
the space of smooth $\mathbb{R}$-valued functions on $X^*$;
hereafter we call $X^*$ a Poisson manifold 
(a point in $X^*$ is denoted by $\bm{\xi}$ ),
and $G(\bm{\xi}) \in C^\infty(X^*)$ a \emph{physical quantity}.
The most important example of a physical quantity is the energy=Hamiltonian.

For $G(\bm{\xi}) \in C^\infty(X^*)$,
we define its \emph{gradient} $\partial_{\bm{\xi}} G ~(\in X)$ by
\[
G(\bm{\xi}+\epsilon\tilde{\bm{\xi}}) -  G(\bm{\xi})
= \epsilon \langle \partial_{\bm{\xi}} G, \tilde{\bm{\xi}} \rangle
+ O(\epsilon^2),
\quad \forall \tilde{\bm{\xi}} \in X^*.
\]
If $\bm{v}\in X$ is given as $\bm{v}=\partial_{\bm{\xi}} H$ with a \emph{Hamiltonian} $H(\bm{\xi}) \in C^\infty(X^*)$, 
its coadjoint action reads Hamilton's equation:
\begin{equation}
\ddt {\bm{\xi}}  = \ad_{\bm{v}}^* \bm{\xi} = [\partial_{\bm{\xi}} H, \bm{\xi} ]^* .
\label{Hamilton-1}
\end{equation}
For a general physical quantity $G(\bm{\xi})\in C^\infty(X^*)$, we may calculate
\begin{equation}
\ddt {G}(\bm{\xi}(t)) = \langle \partial_{\bm{\xi}} G, \ddt {\bm{\xi}} \rangle = \langle \partial_{\bm{\xi}} G, [\partial_{\bm{\xi}} H, \bm{\xi} ]^*\rangle.
\label{Hamilton-2}
\end{equation}
We write the right-hand side as $\{ G, H \}$, and call it the \emph{Lie-Poisson bracket},
i.e.
on the space $ C^\infty(X^*)$ of physical quantities, we define a Poisson algebra by 
\begin{equation}
\{ G , H \}
= \langle \partial_{\bm{\xi}} G, [\partial_{\bm{\xi}} H, \bm{\xi} ]^*\rangle
= \langle  [\partial_{\bm{\xi}} G , \partial_{\bm{\xi}} H ] , \bm{\xi} \rangle.
\label{Lie-Poisson-1}
\end{equation}
The bi-linearity, antisymmetry, and the Leibniz property are evident.
Jacobi's identity inherits that of the Lie bracket $[~,~]$ (see \cite{Morrison1998}).

Denoting
\begin{equation}
J(\bm{\xi}) \circ  = [\, \circ \, , \bm{\xi}]^* \quad : \, X \rightarrow X^*,
\label{Poisson_operator}
\end{equation}
which we call a \emph{Poisson matrix} (or \emph{Poisson operator}, particularly  if $X$ is an infinite-dimensional space),
we may write (\ref{Lie-Poisson-1}) as
\begin{equation}
\{ G,H \} = \langle \partial_{\bm{\xi}} G , J(\bm{\xi}) \partial_{\bm{\xi}} H \rangle .
\label{Lie-Poisson-3}
\end{equation}
Invoking the structure constants $c_{jk}^\ell$ of the Lie bracket $[~,~]$
(see (\ref{structure_constant})), we may write
the Poisson matrix as
\begin{equation}
J(\bm{\xi})_{jk}  = c_{jk}^\ell \xi_\ell ,
\label{Poisson_operator-2}
\end{equation}
whence 
\begin{equation}
\{ G,H \} =  (\partial_{\xi_j}G ) J(\bm{\xi})_{jk} (\partial_{\xi_k}H ) .
\label{Lie-Poisson-4}
\end{equation}

\begin{remark}[related Poisson brackets]
\label{remark:generalized_Lie-Poisson}
\normalfont
One may define a \emph{homogeneous} Poisson bracket such that
\begin{equation}
\{ G,H \} = \langle [\partial_{\bm{\xi}} G , \partial_{\bm{\xi}} H ] , \bm{\phi} \rangle
\label{Lie-Poisson-G}
\end{equation}
with an arbitrary constant vector $\bm{\phi} \in X^*$.
Here,  the Poisson matrix $J(\bm{\phi}) \in \mathrm{Hom}(X,X^*)$ is a homogeneous (constant coefficient) map.
The simplest choice is $\bm{\phi}=\bm{e}^\ell$ ($\ell$ is a fixed index), which gives
\[
\langle [\partial_{\bm{\xi}} G , \partial_{\bm{\xi}} H], \bm{e}^\ell \rangle
= \frac{\partial G}{\partial \xi_j} \frac{\partial H}{\partial \xi_k} \langle [\bm{e}_j,\bm{e}_k], \bm{e}^\ell \rangle
= \frac{\partial G}{\partial \xi_j} \frac{\partial H}{\partial \xi_k} c^\ell_{jk} .
\]
We easily find that the bracket of (\ref{Lie-Poisson-G}) satisfies Jacobi's identity.
We will encounter such brackets when we linearize Lie-Poisson brackets (see Sec.\,\ref{subsec:linearization_around_sing-3D}).
Another interesting idea is the \emph{deformation} of the Lie-Poisson bracket such that
\begin{equation}
\{ G,H \} = \langle [\partial_{\bm{\xi}} G , \partial_{\bm{\xi}} H ] , M \bm{\xi} \rangle ,
\label{Lie-Poisson-G2}
\end{equation}
where $M$ is a certain linear map $X^*\rightarrow X^*$. 
Rewriting (\ref{Lie-Poisson-G2}) as
\begin{equation}
\{ G,H \} = \langle M^{\mathrm{T}} [\partial_{\bm{\xi}} G , \partial_{\bm{\xi}} H ] , \bm{\xi} \rangle ,
\label{Lie-Poisson-G3}
\end{equation}
we may view this as a deformation of the Lie algebra by modifying the bracket form 
$[~,~]$ to $[~,~]_M=M^{\mathrm{T}}[~,~]$.
Of course, there is a strong restriction on $M$ so that the new bracket $[~,~]_M$ satisfies Jacobi's identity.
This is, indeed, the central issue of the following discussions (see Sec.\ref{subsec:deformation}).
\end{remark}

\subsection{The Casimir}

Given a Hamiltonian $H \in C^\infty(X^*)$,
the dynamics of $\bm{\xi}$ (in the Poisson manifold $X^*$) is described by Hamilton's equation:
repeating (\ref{Hamilton-1}) with notation (\ref{Poisson_operator}), 
\begin{equation}
\ddt {\bm{\xi}}  
= J(\bm{\xi}) \partial_{\bm{\xi}} H(\bm{\xi}) ,
\label{Hamilton-1'}
\end{equation}
Let us look at the equilibrium points.
If $J(\bm{\xi})$ is regular (i.e. $\mathrm{Ker}\,J(\bm{\xi})=\{0\}$ for every $\bm{\xi}\in X^*$), 
an  equilibrium point of the Hamiltonian system  (\ref{Hamilton-1'}) must be a critical point of the Hamiltonian.
However, nontrivial $\mathrm{Ker}\,J(\bm{\xi})$ enriches the set of equilibrium points.
If $C(\bm{\xi}) \in C^\infty(X^*)$ satisfies
\begin{equation}
\{ C, G \}=0,
\quad \forall G \in C^\infty (X^*) ,
\label{Casimir_definition}
\end{equation}
we call $C$ a Casimir (or a center element of the Lie-Poisson algebra).
By the definition of the Lie-Poisson bracket, (\ref{Casimir_definition}) is equivalent to
\begin{equation}
J(\bm{\xi}) \partial_{\bm{\xi}} C = 0,
\label{Casimir_definition-2}
\end{equation}
which implies that $C$ is the ``integral'' of an element of $ \mathrm{Ker}\,J(\bm{\xi})$, i.e.
\[
\partial_{\bm{\xi}} C \in  \mathrm{Ker}\,J(\bm{\xi}).
\]
The dynamics governed by (\ref{Hamilton-1'}) is invariant under the transformation
\[
H \mapsto F = H + \mu C \quad (\forall \mu\in\mathbb{R}).
\]
If there are multiple Casimirs, we may include them to define further transformed Hamiltonians
$F=H +\mu_1 C_1 + \mu_2 C_2 + \cdots$.
We call $F$ an \emph{energy-Casimir function}. 
While the critical points of $H$ are often trivial, those of $F$ may have various interesting structure.
As mentioned in the  Introduction, 
such degeneracy  enables even simple Hamiltonians to generate nontrivial structure or dynamics in a 
system dictated by a particular Poisson algebra; the key factor being  the nature of the degeneracy of $J(\bm{\xi})$.

A point where $\mathrm{Rank}\, J(\bm{\xi})$ changes is a \emph{singularity} of the Poisson algebra
(see Remark\,\ref{remark:symplectic}).
By the Lie-Darboux theorem\,\cite{Morrison1998},
the Casimirs foliate the phase space $X^*$, so that, in a neighborhood of every regular point
(where  $\mathrm{Rank}\, J(\bm{\xi})$ is constant), the leaf is locally symplectic,
i.e., there is a local coordinate system in which $J(\bm{\xi})$ is transformed into a standard form
\begin{equation}
J_{D}  = J_c \oplus^\nu 0,
\quad
J_c = \left( \begin{array}{cc}
0  &  ~ I~ 
\\
-I &  0  
\end{array} \right) ,
\label{normal-J}
\end{equation}
where $\nu$ is the nullity of $J(\bm{\xi})$.

In the following discussion, the \emph{singularity} $\sigma = \{ \bm{\xi}_s \in X^*;\, J(\bm{\xi}_s)=0 \}$ will play an important role (see Remark\,\ref{remark:symplectic}).
Needless to say, every point $\bm{\xi}_s\in\sigma$ is an equilibrium point, which we call a \emph{singular equilibrium}, 
and distinguish it from the critical points of the energy-Casimir functional;
the latter will be called \emph{regular equilibria}.

\begin{remark}[symplectic foliation and singularity]
\label{remark:symplectic}
\normalfont
Here we study chirality from the algebraic point of view,  limited to Lie-Poisson systems, and emphasize the  rank changing singularities of the Poisson tensor, as opposed to addressing the bigger geometric picture based on the symplectic foliations possessed by all Poisson manifolds\,\cite{vaisman,weinstein98}.  To understand what we mean by rank changing singularity let  $M$ be a Poisson manifold endowed with a Poisson bracket $\{ G, H \} = \langle \partial_{\bm{\xi}} G, J(\bm{\xi}) \partial_{\bm{\xi}} H \rangle$; here $M=X^*$ (phase space).
Let us denote by ${S}_{\bm{\xi}}$ the totality of the Hamiltonian vectors 
$\{ \bm{\xi}, H \}$ ($\forall H\in C^\infty(M)$) evaluated at the point $\bm{\xi}\in M$; 
here $\{ \bm{\xi}, H \} = \ad_{\bm{h}}^* \bm{\xi} = [\bm{h}, \bm{\xi}]^*$ 
($\forall \bm{h}=\partial_{\bm{\xi}} H$).
In general, ${S}_{\bm{\xi}}$ is a subspace of $T_{\bm{\xi}}$.
The vector bundle ${S} M =\{{S}_{\bm{\xi}};\, \bm{\xi}\in M \}$ is a  \emph{distribution}.
The dimension $r(\bm{\xi})$ of ${S}_{\bm{\xi}}$ is a lower semicontinuous function of $\bm{\xi}\in M$.
Evidently,  $r(\bm{\xi})=\mathrm{Rank}\, J(\bm{\xi}) $.
The \emph{regular point} $\bm{\xi}$ is  where $r(\bm{\xi})=$ constant (local maximum) in the neighborhood of $\bm{\xi}$.
The set $\rho$ of regular points is an open set (not necessarily connected), and   
$\sigma = M \setminus\rho$ is the \emph{singularity}.
In the neighborhood of $\bm{\xi} \in \rho$, the Hamiltonian vectors foliate $M$ into symplectic leaves 
(a leaf will be denoted by $L$).
The Casimir $C$ (if it exists) is an integral of the kernel (null  space) of $J(\bm{\xi})$ (i.e., $J(\bm{\xi})\partial_{\bm{\xi}}C =0$),
implying that $C$ is constant on each leaf  ($i^*_L \rmd C =0$),
or the exact 1-form $\rmd C$ is the normal vector on  the leaves.
The singularity $\sigma$ is detected as the set of points where $C$ (or $\rmd C$) becomes singular.
Therefore, the Casimirs are useful elements of the Poisson algebra to characterize the foliation.
\\
\indent In general, however, the kernel is not necessarily integrable, and the nullity $=\mathrm{dim}\,M - r(\bm{\xi})$ can be larger than the number of independent Casimirs;
this is the ``Casimir deficit'' problem, and is typical at the singularity\,\cite{Morrison1998}.
We also note that some Casimirs fail to identify (parameterize) the symplectic foliation;
this occurs when a non-compact leaf is immersed densely in the phase space (a  Kronecker foliation); then clearly there cannot be an appropriate nonconstant function which is constant on all leaves\,\cite{weinstein97,zakrzewski}.  
In the present work, we only concentrate on the singularities where $r(\bm{\xi})$ goes to zero,
but there are more moderate class of singularities where $r(\bm{\xi})$ drops to some finite number.
The ``interior'' of such singularities still maintains dynamics\,\cite{YM2016}.
\\
\indent
In a broad sense,  our  study of chirality is related to  singular symplectic foliations, including $b$-Poisson or log-symplectic manifolds and  symplectic manifolds with boundary (see e.g.\ \cite{eva}). 
In higher dimensions  there are other kinds of singularities when the rank  drops by 2, 4, 6,. . .,  and besides spectra other interesting ``non-Hamiltonian" behavior is possible.  Such study is beyond the scope of the present paper.

\end{remark}

\subsection{Hamiltonian spectral symmetry}
\label{subsec:linearization}

It is known that the linearization of Hamilton's equation (\ref{Hamilton-1}) around a regular equilibrium point
(to be denoted by $\bm{\xi}_r$) yields a linearized Hamiltonian system,
and that spectra of the linearized generator have the \emph{Hamiltonian symmetry} (with a zero eigenvalue of multiplicity $\nu$);
if $\lambda$ is an eigenvalue of the generator, $-\lambda$ is also an eigenvalue  (implying a time-reversal symmetry),
and $\overline{\lambda}$ is also an eigenvalue  (see Remark\,\ref{remark:Hamiltonian_symmetry}).
Remember that our linearized PRS system (\ref{SeqEom}) does not obey this theorem,
even though the original PRS system (\ref{MT-1}) is Hamiltonian.
This is because (\ref{SeqEom}) is a linearization around a singular equilibrium point.
Let us see how the linearization works out differently for regular and singular equilibrium points.


First consider general  equilibrium points of Hamilton's equation (\ref{Hamilton-1}), 
we use the energy-Casimir functional (if the Poisson algebra has Casimirs) in place of the Hamiltonian,
but we will denote it by $H(\bm{\xi})$ for simplicity.
Around a given equilibrium point $\bm{\xi}_e$ (either regular or singular), 
we consider a small amplitude perturbation $\epsilon\tilde{\bm{\xi}}$ to write $\bm{\xi}=\bm{\xi}_e + \epsilon \tilde{\bm{\xi}}$.
Approximating (\ref{Hamilton-1}) to the first order of $ \epsilon$, we obtain the linearized Hamilton's equation:
\begin{equation}
\frac{\rmd}{\rmd t} \tilde{\bm{\xi}} = J(\bm{\xi}_e) H''(\bm{\xi}_e) \tilde{\bm{\xi}}  + J(\tilde{\bm{\xi}}) \bm{h} (\bm{\xi}_e),
\label{Hamilton-L}
\end{equation}
where 
\[
\bm{h} (\bm{\xi}_e) = (\partial_{\bm{\xi}} H)|_{\bm{\xi}=\bm{\xi}_e} \quad \in X
\]
is the Hamiltonian vector evaluated at the equilibrium point $\bm{\xi}_e$, and 
\[
( H''(\bm{\xi}_e)_{jk} ) = (\partial_{\xi_k} \partial_{\xi_j} H)|_{\bm{\xi}=\bm{\xi}_e} \quad \in \mathrm{Hom}(X^*,X)
\]
is the Hessian of $H(\bm{\xi})$ evaluated at $\bm{\xi}_e$.
In deriving $J(\bm{\xi}_e+\tilde{\bm{\xi}})-J(\bm{\xi}_e) = J(\tilde{\bm{\xi}})$,
we have used the fact that $J(\bm{\xi})$ is a linear function of $\bm{\xi}$ (see (\ref{Poisson_operator})).

Interestingly, different equilibrium points pick up different terms from the right-hand side of (\ref{Hamilton-L}),
depending on whether they are regular or singular:
\begin{itemize}
\item
\textbf{Regular equilibrium:}
The regular equilibrium $\bm{\xi}_r$ is a point where $\partial_{\bm{\xi}} H=0$;
hence the second term on the right-hand side of (\ref{Hamilton-L}) vanishes.
The Hamiltonian (obtained by expanding the energy-Casimir functional) of a perturbation is
\begin{equation}
H_L(\tilde{\bm{\xi}}) = \frac{1}{2} \langle {H''}(\bm{\xi}_e) \tilde{\bm{\xi}},  \tilde{\bm{\xi}}  \rangle ,
\label{Hamilton-L-ec-Hamiltonian}
\end{equation}
by which we may cast  (\ref{Hamilton-L}) into a Hamiltonian form:
\begin{equation}
\frac{\rmd}{\rmd t} \tilde{\bm{\xi}} = J(\bm{\xi}_r) \partial_{\tilde{\bm{\xi}}} H_L(\tilde{\bm{\xi}}) .
\label{Hamilton-L-ec-2}
\end{equation}
Notice that $J(\bm{\xi}_r) $ is the Poisson operator $J(\bm{\xi})$ evaluated at the fixed equilibrium point,
which defines a homogeneous Poisson algebra (see Remark\,\ref{remark:generalized_Lie-Poisson}).
The spectra of these equilibria  have the Hamiltonian symmetry (Remark\,\ref{remark:Hamiltonian_symmetry}).

\item
\textbf{Singular equilibrium:}
The singular equilibrium $\bm{\xi}_s$ is a point where $J(\bm{\xi}_s)=0$.
Then, the first term on the right-hand side of (\ref{Hamilton-L}) vanishes, and we have 
\begin{equation}
\frac{\rmd}{\rmd t} \tilde{\bm{\xi}} = J(\tilde{\bm{\xi}}) \bm{h} (\bm{\xi}_s)
=  [ \bm{h}(\bm{\xi}_s) , \tilde{\bm{\xi}} ]^*.
\label{Hamilton-LS}
\end{equation}
Here $\bm{h}(\bm{\xi}_s)$ is a fixed vector, so that the dynamics stems from $J(\tilde{\bm{\xi}})$.
There is no guarantee that (\ref{Hamilton-LS}) is a Hamiltonian system.
Indeed, the linearized PRS system (\ref{SeqEom}) is of this type, 
in which the Hamiltonian symmetry is broken\,\cite{YTM-rattleback}.
However, we also find that the linearized generator $\mathcal{A} = [ \bm{h}(\bm{\xi}_s) , \circ ]^*$ becomes a Hamiltonian vector field for a special class of Lie brackets.
We will identify such class of Lie algebras in  the next section.
\end{itemize}

While the singular linearized system (\ref{Hamilton-LS}) is not necessarily Hamiltonian, we have the following conservation laws.

\begin{proposition}[conservation laws] 
\label{proposition:conservation}
The linear system (\ref{Hamilton-LS}) has the following invariants:
\begin{enumerate}
\item
The Casimir $\mathcal{C}(\tilde{\bm{\xi}})$ of the original nonlinear system (\ref{Hamilton-1}),
which is evaluated for the perturbation $\tilde{\bm{\xi}}$.
\item
The first-order energy $H_1(\tilde{\bm{\xi}}) = \langle \bm{h}(\bm{\xi}_s), \tilde{\bm{\xi}} \rangle$.

\end{enumerate}
\end{proposition}

\begin{proof}
For a Casimir $\mathcal{C}(\tilde{\bm{\xi}})$, 
(\ref{Casimir_definition-2}) implies
\begin{eqnarray*}
\frac{\rmd}{\rmd t} \mathcal{C}(\tilde{\bm{\xi}})
=  \langle \partial_{\tilde{\bm{\xi}}} \mathcal{C} , \frac{\rmd}{\rmd t} {\tilde{\bm{\xi}}} \rangle
&=&  \langle \partial_{\tilde{\bm{\xi}}} \mathcal{C} , J(\tilde{\bm{\xi}})  \bm{h}(\bm{\xi}_s) \rangle
\\
&=&  \langle [\partial_{\tilde{\bm{\xi}}} \mathcal{C} , \bm{h}(\bm{\xi}_s)],  \tilde{\bm{\xi}} \rangle
\\
&=&  - \langle [ \bm{h}(\bm{\xi}_s), \partial_{\tilde{\bm{\xi}}} \mathcal{C}],  \tilde{\bm{\xi}} \rangle
\\
&=& - \langle \bm{h}(\bm{\xi}_s) , J(\tilde{\bm{\xi}}) \partial_{\tilde{\bm{\xi}}} \mathcal{C} \rangle =0.
\end{eqnarray*}
We also observe,  by the anti-symmetry of $J$, 
\[
\frac{\rmd}{\rmd t} {H}_1(\tilde{\bm{\xi}})
=  \langle \bm{h}(\bm{\xi}_s) , \frac{\rmd}{\rmd t} {\tilde{\bm{\xi}}} \rangle
=  \langle \bm{h}(\bm{\xi}_s) , J(\tilde{\bm{\xi}}) \bm{h}(\bm{\xi}_s)  \rangle =0.
\]
\begin{flushright}
~~\qed
\end{flushright}
\end{proof}

Note that the first-order energy $H_1(\tilde{\bm{\xi}}) $ is different from the second-order energy $H_L( \tilde{\bm{\xi}}) $, 
given by (\ref{Hamilton-L-ec-Hamiltonian}),
which is an invariant of the regular linearized system (\ref{Hamilton-L-ec-2}).

\begin{remark}[Hamiltonian symmetry of spectra]
\label{remark:Hamiltonian_symmetry}
\normalfont
The generator (matrix) of a linear Hamiltonian system can be written as
\[
\mathcal{A} = J \mathcal{H},
\]
where $J$ is a constant coefficient Poisson matrix,
and $\mathcal{H}$ is a constant coefficient symmetric matrix
that  is  the Hessian of some Hamiltonian.
As seen above, the  quadratic form $H_L(\tilde{\bm{\xi}}):=\frac{1}{2} \langle \mathcal{H}\tilde{\bm{\xi}},\tilde{\bm{\xi}}\rangle$ is 
the  Hamiltonian for the linear Hamiltonian system
\[
\frac{\rmd}{\rmd t} {\tilde{\bm{\xi}}} = J \partial_{\tilde{\bm{\xi}}} H_L = \mathcal{A} \tilde{\bm{\xi}}.
\]
The  eigenvalues of the generator $\mathcal{A}$, which are the solutions of the characteristic equation
\begin{equation}
P(\lambda) := \mathrm{det}\,(\mathcal{A}-\lambda I) = 0,
\label{characteristic_equation}
\end{equation}
are the subjects of the discussion on the following symmetry.
By \emph{Hamiltonian symmetry}, we mean that every eigenvalue $\lambda$ always has the following counterparts that are simultaneously eigenvalues:
$-\lambda$, $\overline{\lambda}$ (and hence, $-\overline{\lambda}$).
Being $\overline{\lambda}$ a simultaneous eigenvalue is evident.
Suppose that $\mathcal{A} \bm{\zeta} = \lambda \bm{\zeta}$.
Since $\mathcal{A}$ is a real-coefficient matrix, the complex conjugate of this equation reads
$\mathcal{A} \overline{\bm{\zeta}} = \overline{\lambda} \, \overline{\bm{\zeta}}$.
Hence, $\overline{\lambda}$ is an eigenvalue and $\overline{\bm{\zeta}}$ is the corresponding eigenvector.
The \emph{time-reversal symmetry} $-\lambda$ is, however, not so obvious.
First, a Hamiltonian generator satisfies $J_c \mathcal{A}^{\mathrm{T}} J_c = \mathcal{A}$ with the canonical Poisson matrix (symplectic matrix) $J_c$; see (\ref{normal-J}).
Using this, together with $J_c J_c =-I$ and $\mathrm{det} J_c =1$, we observe
\begin{eqnarray*}
P(\lambda) &=& \mathrm{det} (J_c \mathcal{A} ^{\mathrm{T}} J_c + \lambda J_c I J_c )
\\
&=& (\mathrm{det} J_c) [\mathrm{det} (\mathcal{A} ^{\mathrm{T}} + \lambda I )] (\mathrm{det} J_c) 
\\
&=& \mathrm{det} (\mathcal{A}  + \lambda I ) = P(-\lambda).
\end{eqnarray*} 
Hence, $-\lambda$ also satisfies the characteristic equation.
\end{remark}
 
\section{Three-dimensional Lie-Poisson systems}
\label{section:3D}

The Bianchi classification of the 3-dimensional Lie algebras guides us to delineate the mathematical 
structure that leads to \emph{chirality} (the symmetry breaking of Hamiltonian spectra) around the 
singular equilibrium points of Lie-Poisson systems.
We start by reviewing the Bianchi classification.

\subsection{Bianchi classification of 3-dimensional Lie algebras}
\label{subsec:3D}

The real 3-dimensional Lie algebras can be classified by the scheme used to describe the Bianchi cosmologies, 
which divides them into nine types (e.g.\ \cite{Ellis,Ryan}). 
The multiplication tables are given in Table\,\ref{table:Bianchi}.

\begin{table}[tb]
\caption{Three-dimensional Lie algebras in the order of Bianchi classification
(see eg.\ \cite{Ryan}).
Type \rn{1} is abelian, so every $[\bm{e}_j,\bm{e}_k]$ is zero.
Type \rn{2} is the Heisenberg algebra.  
Type \rn{9} is $\mathfrak{so}(3)$.
Type \rn{8} may be regarded as the 3-dimensional Minkowski and describes the Kida vortex \cite{pjmMF}.
Notice that types $\mathrm{\rn{6}}_\eta$ and $\mathrm{\rn{7}}_\eta$ have a parameter $\eta\in\mathbb{R}$.
}
\[
{\small 
\begin{array}{ll} 
\begin{array}{r|ccc}
\makebox[2em][c]{\rn{2}}  & [\circ, \bm{e}_1] & [\circ, \bm{e}_2] & [\circ, \bm{e}_3]  \\
\hline
\bm{e}_1       & 0                           &  0                         &  0                         \\
\bm{e}_2       &  -                           &  0                         &  \bm{e}_1            \\
\end{array}
&
\begin{array}{r|ccc}
\makebox[2em][c]{\rn{3}}  & [\circ, \bm{e}_1] & [\circ, \bm{e}_2] & [\circ, \bm{e}_3]  \\
\hline
\bm{e}_1       & 0                           &  0                         &  \bm{e}_1            \\
\bm{e}_2       &  -                           &  0                         &  0                        \\
\end{array}
\\ ~&~ \\ 
\begin{array}{r|ccc}
\makebox[2em][c]{\rn{4}}  & [\circ, \bm{e}_1] & [\circ, \bm{e}_2] & [\circ, \bm{e}_3]  \\
\hline
\bm{e}_1       & 0                           &  0                         &  \bm{e}_1            \\
\bm{e}_2       &  -                           &  0                         &   \bm{e}_1+\bm{e}_2 \\
\end{array}
&
\begin{array}{r|ccc}
\makebox[2em][c]{\rn{5}}  & [\circ, \bm{e}_1] & [\circ, \bm{e}_2] & [\circ, \bm{e}_3]  \\
\hline
\bm{e}_1       & 0                           &  0                         &  \bm{e}_1            \\
\bm{e}_2       &  -                           &  0                         &  \bm{e}_2            \\
\end{array}
\\ ~&~ \\
\begin{array}{r|ccc}
\makebox[2em][c]{$\RN{6}_{\eta}$}  & [\circ, \bm{e}_1] & [\circ, \bm{e}_2] & [\circ, \bm{e}_3]  \\
\hline
\bm{e}_1       & 0                           &  0                         &  \bm{e}_1            \\
\bm{e}_2       &  -                           &  0                         &  \eta\bm{e}_2 \\
\end{array}
&
\begin{array}{r|ccc}
\makebox[2em][c]{$\RN{7}_{\eta}$}  & [\circ, \bm{e}_1] & [\circ, \bm{e}_2] & [\circ, \bm{e}_3]  \\
\hline
\bm{e}_1       & 0                           &  0                         &  \bm{e}_2            \\
\bm{e}_2       &  -                           &  0                         &  -\bm{e}_1+\eta\bm{e}_2 \\
\end{array}
\\ ~&~ \\ 
\begin{array}{r|ccc}
\makebox[2em][c]{\rn{8}}  & [\circ, \bm{e}_1] & [\circ, \bm{e}_2] & [\circ, \bm{e}_3]  \\
\hline
\bm{e}_1       & 0                           &  \bm{e}_3             &  \bm{e}_2            \\
\bm{e}_2       &  -                           &  0                         &  -\bm{e}_1 \\
\end{array}
&
\begin{array}{r|ccc}
\makebox[2em][c]{\rn{9}}  & [\circ, \bm{e}_1] & [\circ, \bm{e}_2] & [\circ, \bm{e}_3]  \\
\hline
\bm{e}_1       & 0                           &  \bm{e}_3             &  -\bm{e}_2            \\
\bm{e}_2       &  -                           &  0                         &  \bm{e}_1 \\
\end{array}

\end{array} 
}
\]
\label{table:Bianchi}
\end{table}


\begin{table}
\caption{Three-dimensional class-A Lie-Poisson algebras (Bianchi classification).
}
\vspace*{5mm}
\includegraphics[scale=0.6]{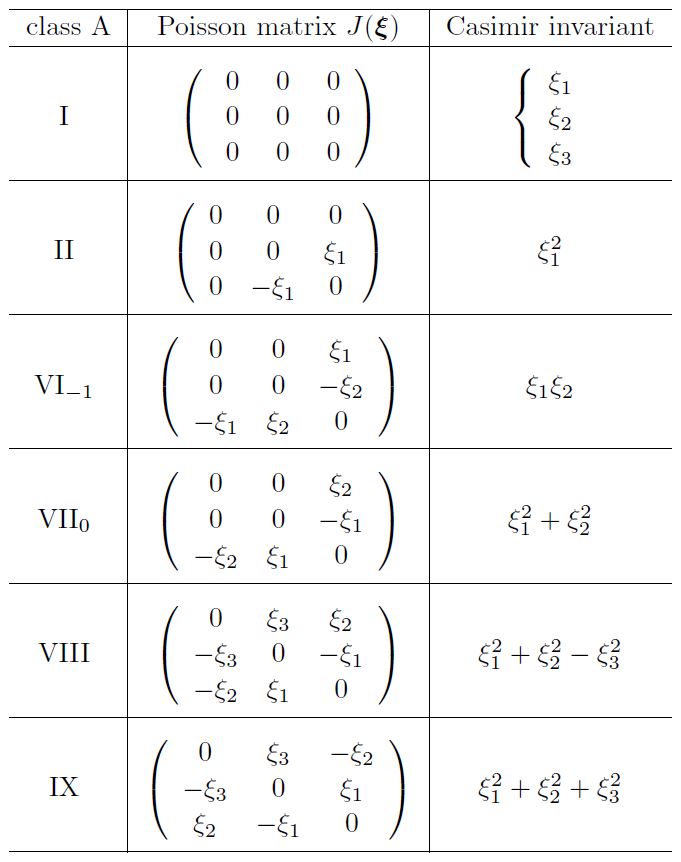}
\label{table:Bianchi-Casimirs-A}
\end{table}

\begin{table}
\caption{Three-dimensional class-B Lie-Poisson algebras (Bianchi classification).
To avoid redundancy, for type-$\mathrm{IV}_\eta$, $\eta\neq 0,1$.   The Casimir of type $\mathrm{VII}_{\eta\neq0}$ 
needs further classification:
$\eta^2>4$ gives
(denoting $\lambda_\pm=(-\eta\pm\sqrt{\eta^2-1})/2$)
 $C_{\mathrm{VII}_{\eta\neq 0}}=\lambda_- \log(-\lambda_- \xi_1 -\xi_2) - \lambda_+\log (\lambda_+ \xi_1 +\xi_2)$;
$\eta = \pm 2$ gives
$C_{\mathrm{VII}_{\eta\neq 0}}= \frac{\pm \xi_2 }{\xi_1 \mp \xi_2}+\log(\xi_1\mp \xi_2)$;
$\eta^2 < 4$ gives (putting $a = -\eta/2$ and $\omega=\sqrt{1-\eta^2/4}$, i.e.\ $\lambda_\pm = a\pm i\omega$)
$C_{\mathrm{VII}_{\eta\neq 0}}= 2 a \arctan \frac{a\xi_1 + \xi_2}{\omega \xi_1} -
\omega \log [ (a\xi_1 + \xi_2)^2 + (\omega \xi_1)^2 ] $.
}
\vspace*{5mm}
\includegraphics[scale=0.6]{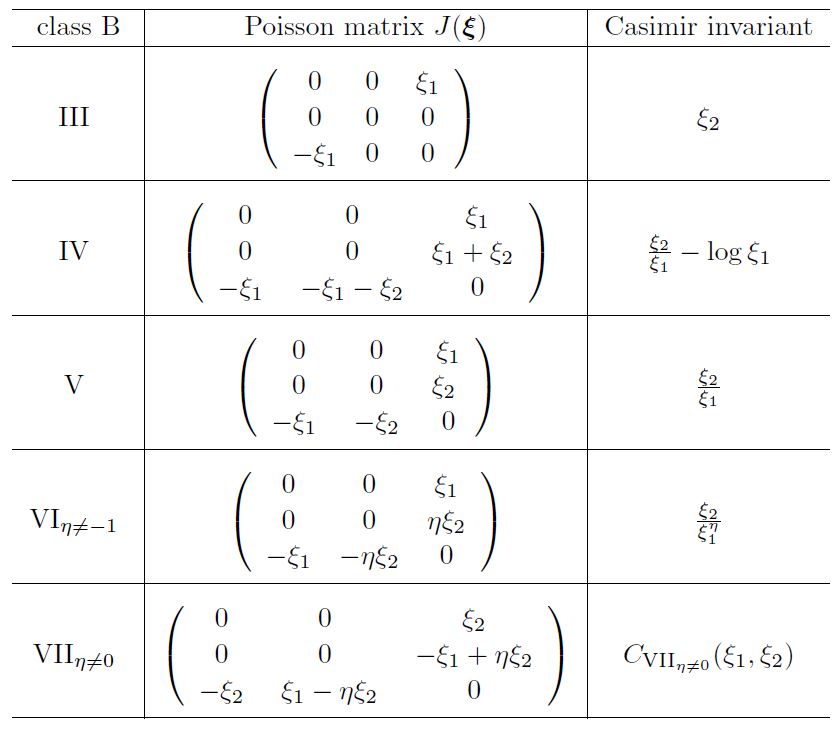}
\label{table:Bianchi-Casimirs-B}
\end{table}

Tables\,\ref{table:Bianchi-Casimirs-A} and \ref{table:Bianchi-Casimirs-B} summarize the 
$3\times3$ Poisson matrices $J$ defined as
\[
J_{jk} = c^\ell_{jk} \xi_\ell ,
\]
which gives the Lie-Poisson brackets $\{G,H \} = \langle \partial_{\bm{\xi}} G, J \partial_{\bm{\xi}}H\rangle$ as in 
  (\ref{Poisson_operator-2}).

Among the nine possibilities, type-IX corresponds to $\mathfrak{so}(3)$, 
which gives the Lie-Poisson matrix
\begin{equation}
J_{\mathrm{IX}}(\bm{\xi}) \bm{u} = \left( \begin{array}{ccc}
0 &  \xi_3 & -\xi_2 \\
-\xi_3 & 0 & \xi_1 \\
\xi_2 & -\xi_1 & 0
\end{array} \right) 
\left( \begin{array}{c} u^1 \\ u^2 \\ u^3 \end{array} \right) 
=
-\bm{\xi}\times \bm{u} .
\label{so(3)}
\end{equation}
For example, the Euler top obeys Hamilton's equation
$\frac{\rmd}{\rmd t}{\bm{\xi}} = J_{\mathrm{IX}}(\bm{\xi}) \partial_{\bm{\xi}} H $ with a Hamiltonian
$H = \sum_{j=1}^3 \xi_j ^2 /I_j $ ($I_j$ being  the inertial moment along the axis $\bm{e}^j$)\,\cite{Morrison1998}.

As announced in Sec.\,\ref{sec:PRS}, the PRS system of the  rattleback is a type-VI system;
with the $\mathrm{\rn{6}}_\eta $ Poisson matrix
$J_{\mathrm{\rn{6}}}(\bm{\xi})$ and a symmetric Hamiltonian $H(\bm{\xi})= \| \bm{\xi} \|^2/2$,
Hamilton's equations $\frac{\rmd}{\rmd t}{\bm{\xi}} = J_{\mathrm{\rn{6}}}(\bm{\xi}) \partial_{\bm{\xi}} H $,
under the correspondences $\xi_1=P$, $\xi_2=R$, $\xi_3=S$, 
and $\eta=-\alpha$ reproduces the PRS system (\ref{MT-1}).

\subsection{Class-A and class-B}
\label{subsec:class-A_class-B}

 The Bianchi types are divided into two classes: class A, composed of types $\mathrm{I}, \mathrm{II}, \mathrm{VI}_{-1}, \mathrm{VII}_0, \mathrm{VIII}$, and $\mathrm{IX}$,  and class B, composed of types $\mathrm{III}, \mathrm{IV}, \mathrm{V}, \mathrm{VI}_{\eta\neq-1}$, and $\mathrm{VII}_{\eta\neq 0}$. 
Somewhat fortuitously, this  classification turns out to separate non-chiral systems from chiral:
 the class-A systems maintain the Hamiltonian symmetry, while class-B systems have chiral spectra.
 Before analyzing the reason for this, we summarize some direct observations.

In Tables\,\ref{table:Bianchi-Casimirs-A} and \ref{table:Bianchi-Casimirs-B}, 
we list the Casimirs for the Lie-Poisson brackets associated with each algebra.  
 We find that the Casimirs of class-A Lie-Poisson brackets (Table\,\ref{table:Bianchi-Casimirs-A}) are all quadratic forms, 
while those of class-B (Table\,\ref{table:Bianchi-Casimirs-B}) are ``singular'' functions.
Therefore, the class-A Casimir leaves are algebraic varieties, each of which defines a two-dimensional symplectic manifold (see Fig.\,\ref{fig:Bianchi_leaves-A}).
 On the other hand, every class-B leaf contains a singularity of some kind, so is only locally symplectic (see Fig.\,\ref{fig:Bianchi_leaves-B}).

\begin{figure}
\begin{center}
\includegraphics[scale=0.5]{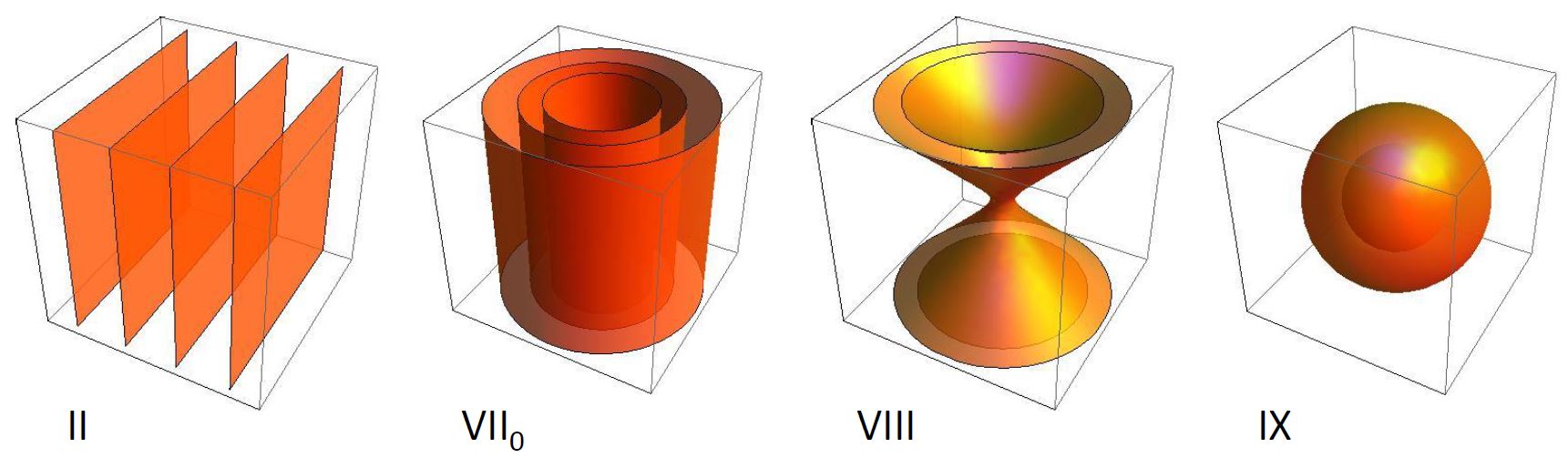}
\caption{
\label{fig:Bianchi_leaves-A}
The foliated phase spaces of the class-A Bianchi Lie-Poisson algebras.
The leaves are level sets of the Casimirs given in Table\,\ref{table:Bianchi-Casimirs-A}.  
}
\end{center}
\end{figure}

\begin{figure}
\begin{center}
\includegraphics[scale=0.65]{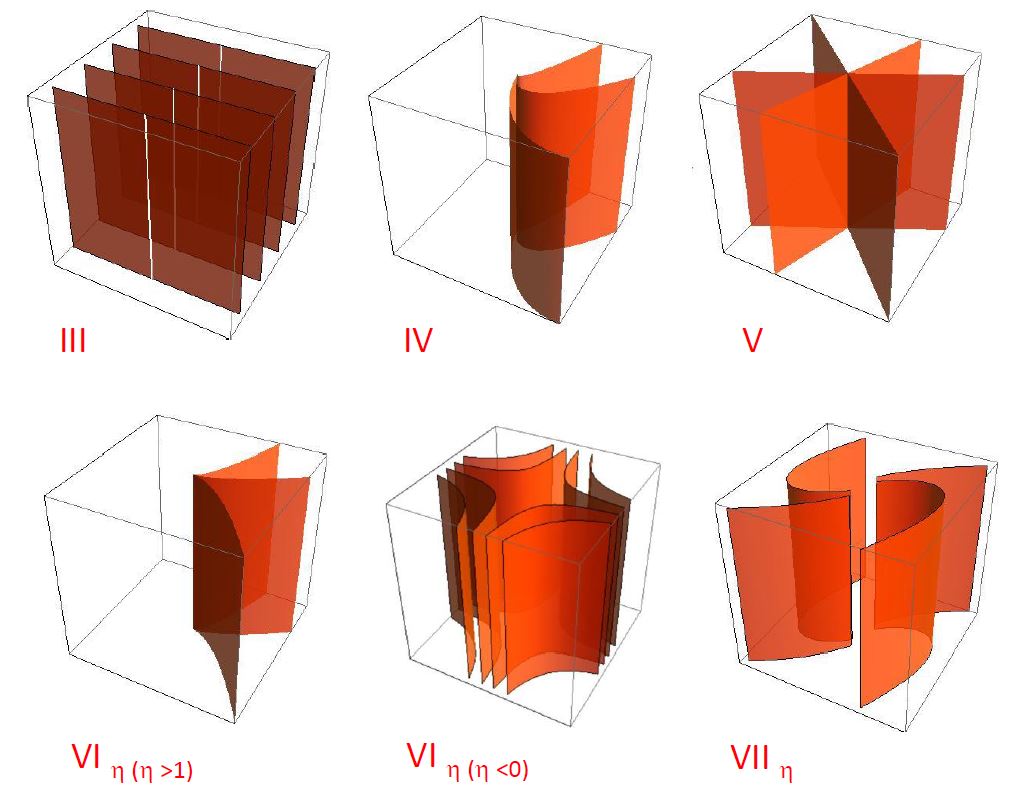}
\caption{
\label{fig:Bianchi_leaves-B}
The foliated phase spaces of the class-B Bianchi Lie-Poisson algebras.
The leaves are level sets of the Casimirs given in Table\,\ref{table:Bianchi-Casimirs-B}.  
}
\end{center}
\end{figure}
 

\subsection{Spectra of class-A and class-B Lie-Poisson systems around singularities}
\label{subsec:3D_spectra}

Let us calculate the spectra of the linearized systems around singular equilibrium points.
We exclude the trivial ($J(\bm{\xi})\equiv0$) type-I system. 
The singularity (the set of singular equilibrium points:
$\sigma = \{ \bm{\xi}_s \in X^*;\, J(\bm{\xi}_s)=0 \}$) varies from two-dimensional to zero-dimensional:
\begin{itemize} 
\item
$\sigma = \{ \bm{\xi}=(0~\xi_2 ~ \xi_3)^{\mathrm{T}} \}$
for type II and III.
\item
$\sigma = \{ \bm{\xi}=(0~\,0\,~\xi_3)^{\mathrm{T}}  \}$
from type IV through VII.
\item
$\sigma = \{ \bm{\xi}=(0~\,0\,~\,0\,)^{\mathrm{T}}  \}$
for type VIII and IX.
\end{itemize} 
The generator of the linear system is, as given in (\ref{Hamilton-LS}),
\[
( \mathcal{A}^\ell_j  ) = [ \bm{h}(\bm{\xi}_s), \circ ]^* = (c^\ell_{jk }h^k ),
\]
where $h^k = \partial_{\xi_k} H |_{\bm{\xi}_s}$.
In 3-dimensional systems, the orbit is given by the intersection of the levels of the Casimir $C(\tilde{\bm{\xi}})$ and the linearized energy $H_1(\tilde{\bm{\xi}})$  (Proposition\,\ref{proposition:conservation}).
Because of the foliation by $C(\tilde{\bm{\xi}})$, one of the eigenvalues of the linearized system (\ref{Hamilton-LS}) must be zero.
For the spectrum to be Hamiltonian,
the remaining two eigenvalues must be either a pair $\pm i\omega$ of imaginary numbers
or a pair $\pm\gamma$ of real numbers.
Therefore, the Hamiltonian symmetry implies a time-reversal symmetry.

Tables\,\ref{table:Bianchi-linear-A} and \ref{table:Bianchi-linear-B} summarize the spectra of each linearized system.
It is evident that the spectra of the class-A systems (Table\,\ref{table:Bianchi-linear-A}) have the Hamiltonian symmetry, while those of the class-B systems (Table\,\ref{table:Bianchi-linear-B}) do not.
For example, class-B, type-VI with $h^3=\partial_{\xi_3}H$ ($H=\|\bm{\xi}\|^2/2$) reproduces the rattleback chiral spectra under the correspondences $\xi^3=S$ and $\eta=-\alpha$ (see (\ref{SeqEom})).

As mentioned in Sec.\,\ref{subsec:linearization}, the linearization about a singular equilibrium point does not yield a (linear) Hamiltonian system (unlike the linearization about a regular equilibrium point),
so it is more surprising that the class-A systems do have Hamiltonian spectra than that the class-B systems break the Hamiltonian symmetry.
The reason for these correspondences will be elucidated  in Sec.\,\ref{sec:deformation}.

\begin{table}
\caption{Linearized class-A systems around a singular equilibrium point $\bm{\xi}_s$.
The singularity $\sigma = \{ \bm{\xi};\, J(\bm{\xi})=0 \}$,
the generator $\mathcal{A}=[\bm{h}(\bm{\xi}_s), \,\circ\,]^*$, and the characteristic equation 
$\mathrm{det}(\lambda I -\mathcal{A})=0$ of each class-A Bianchi Lie-Poisson system are summarized
(those of class-B systems are given in Table\,\ref{table:Bianchi-linear-B}).
We denote $h^j =\partial_{\xi_j} H|_{\bm{\xi}_s}$.
Type-I algebra is omitted, because it is abelian so that the Poisson bracket is trivial. 
}
\vspace*{5mm}
\begin{center}
\includegraphics[scale=0.6]{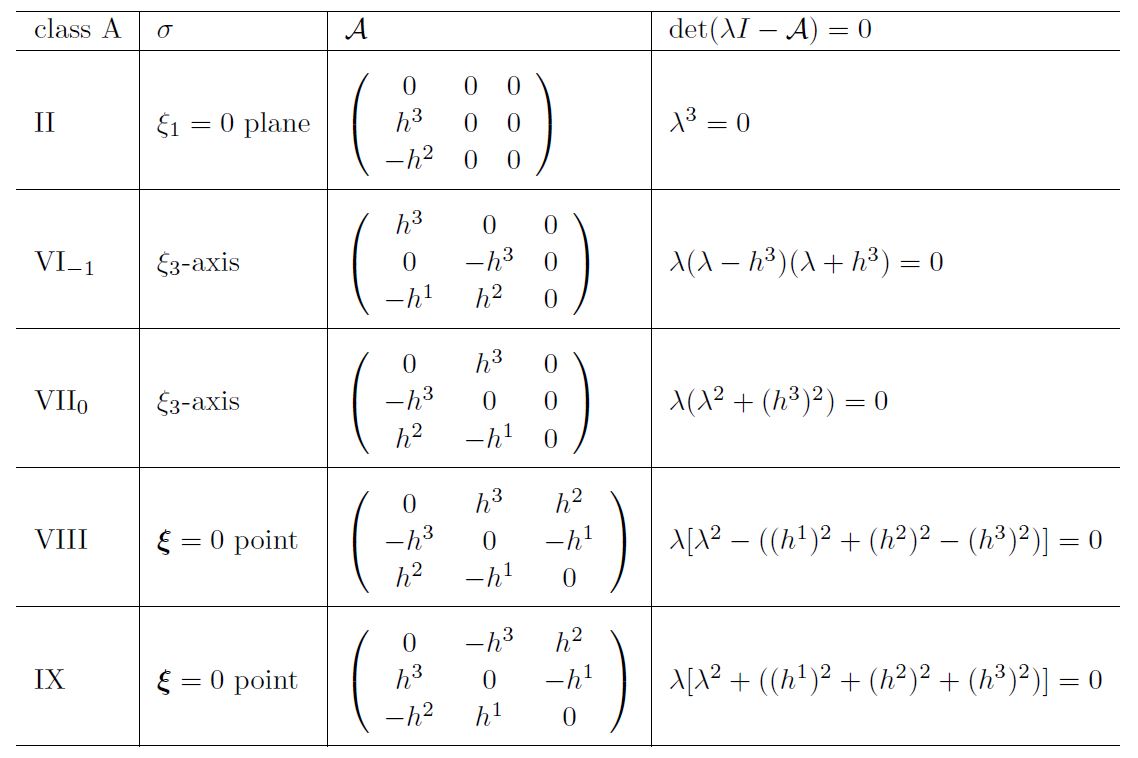}
\end{center}
\label{table:Bianchi-linear-A}
\end{table}

\begin{table}
\caption{Linearized class-B systems around a singular equilibrium point $\bm{\xi}_s$.
}
\vspace*{5mm}
\begin{center}
\includegraphics[scale=0.6]{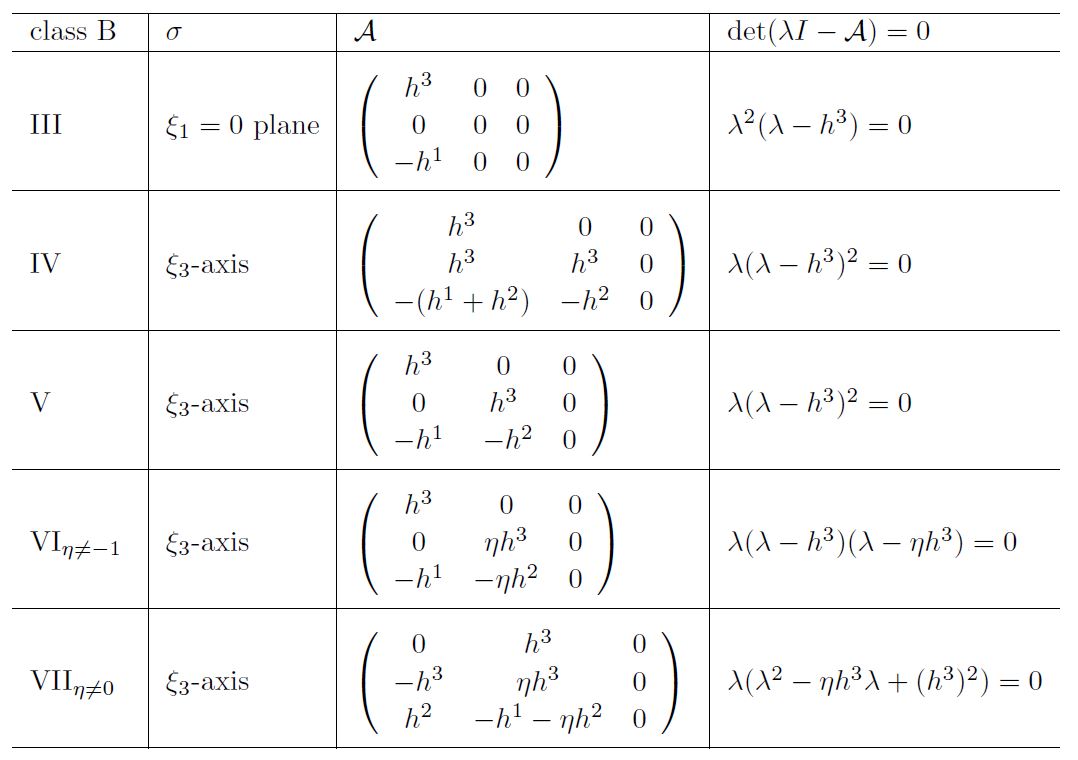}
\end{center}
\label{table:Bianchi-linear-B}
\end{table}

\subsection{Geometrical interpretation}
We notice that all Bianchi Lie-Poisson matrices (Tables\,\ref{table:Bianchi-Casimirs-A} and \ref{table:Bianchi-Casimirs-B}) are reversed ($J\mapsto - J$)
by the transformation $\mathscr{T}_3:\,\xi_3 \mapsto -\xi_3$.
Hence, Hamilton's equation of motion (\ref{Hamilton-1}) is invariant
with respect to the time-reversal $\mathscr{T}_t:\,t \mapsto -t$ combined with $\mathscr{T}_3$.
Evidently, all Casimir leaves are invariant with respect to the transformation $\mathscr{T}_3$
(in Figs.\,\ref{fig:Bianchi_leaves-A} and \ref{fig:Bianchi_leaves-B}, $\xi_3$ is the vertical axis).

The linearized systems inherit this time-reversal symmetry of the original Hamiltonian system.
In the context of spectral symmetry, however, there is an additional constraint to be taken into account.
Notice that the transformation $\mathscr{T}_3$ flips the sign of $h^3=\partial_{\xi_3} H$.
With this transformation, all spectra of class-A systems (Table\,\ref{table:Bianchi-linear-A}), 
as well as those of class-B systems (Table\,\ref{table:Bianchi-linear-B}), have the time-reversal symmetry.
However, in the linear theory,
the coefficients included in the generator $\mathscr{A}$ are fixed numbers pertinent to the equilibrium state $\bm{\xi}_s$.
Therefore, in the argument of spectral symmetry, $h^3$ must not be transformed.
Here, class A  contrasts with class B, because of  the existence of alternative transformations.
The Poisson matrices of type-II, type-$\textrm{VII}_0$, type-VIII, and type-IX are reversed by the transformation
$\mathscr{T}_2:\,\xi_2 \mapsto -\xi_2$.
The Poisson matrix of type-$\textrm{VI}_{-1}$ is reversed by
$\mathscr{T}_{12}:\, (\xi_1~\xi_2)  \mapsto (\xi_2~\xi_1)$.
These transformation yield the Hamiltonian (time-reversal symmetric) spectra of the corresponding linearized generators.
To the contrary, the Poisson matrices of class B do not have such symmetry;
evidently, the Casimir invariants of class B algebras are not invariant with respect to $\mathscr{T}_2$ or
$\mathscr{T}_{12}$.

Let us see how the non-Hamiltonian (chiral) spectra are created in the class-B systems.
The existence of the singularity $\sigma$ on the Casimir leaves
(excepting those of type-$\mathrm{VI}_{\eta\neq -1}$ with $\eta<0$, which will be discussed separately)
poses an obstacle for the time-reversal symmetry.
By Proposition~\ref{proposition:conservation}, the orbits are on the levels of Casimirs $C(\tilde{\bm{\xi}})$.
The levels of the linearized energy $H_L(\tilde{\bm{\xi}})$ are planes including $\bm{\xi}_s\in\sigma$;
hence the orbits are connected to the singularities, implying that only real eigenvalues can occur.
The Casimir invariants $C(\xi_1,\xi_2)$ of  class-B systems, however, forbids the co-existence a pair time constants
$\gamma$ and $-\gamma$.

As noted above, type-$\mathrm{VI}_{\eta\neq -1}$ with $\eta<0$ is somewhat special.
Although the singularity $\sigma = \xi_3$-axis is not included in the Casimir leaves,
each level of $C(\bm{\xi})$ is divided into separate surfaces,
preventing circulating orbits around $\bm{\xi}_s\in\sigma$.
Hence, only real eigenvalues can occur, and only  the special value $\eta=-1$   yields  symmetric eigenvalues $\lambda = \pm h^3$.

\section{Deformation of Lie-Poisson algebras}
\label{sec:deformation}

\subsection{Deformation of observables and its reflection to Lie algebras}
\label{subsec:deformation}

The central idea of the following exploration is to characterize the variety of 3-dimensional Lie-Poisson algebras (and their underlying Lie algebras) as \emph{deformations} from a \emph{mother} algebra\,\cite{deformation}.
We will show that the symmetry and asymmetry of the deformations correspond to class A and B.
 
Remembering the argument of Sec.\,\ref{subsec:Lie-Poisson}, it stands to reason that we ask how phenomena will vary, when we modify the observables (cf.\ Remark\,\ref{remark:generalized_Lie-Poisson}).
With $M \in \mathrm{End}(X^*)$, we deform the Lie-Poisson bracket (\ref{Lie-Poisson-1}) as
\begin{equation}
\{ G, H \}_M = \langle [\partial_{\bm{\xi}} G,  \partial_{\bm{\xi}} H], M \bm{\xi} \rangle
 = \langle \partial_{\bm{\xi}} G,  [\partial_{\bm{\xi}} H, M \bm{\xi} ]^* \rangle .
\label{deformation_def}
\end{equation}
Hence, the deformed Poisson matrix (operator) is
\begin{equation}
J_M (\bm{\xi}) = J(M \bm{\xi}).
\label{deformation_def-2}
\end{equation}
With the adjoint matrix (operator) $M^{\mathrm{T}} \in \mathrm{End}(X)$, we may rewrite (\ref{deformation_def}) as
\begin{equation}
\{ G, H \}_M = \langle M^{\mathrm{T}} [\partial_{\bm{\xi}} G,  \partial_{\bm{\xi}} H], \bm{\xi} \rangle 
= \langle [\partial_{\bm{\xi}} G,  \partial_{\bm{\xi}} H]_M , \bm{\xi} \rangle .
\label{deformation_def-3}
\end{equation}
Therefore, we may interpret $\{~,~\}_M$ as the Lie-Poisson bracket produced by the deformed Lie bracket $[~,~]_M = M^{\mathrm{T}} [~,~]$.
For this \emph{deformation} to be allowed, $[~,~]_M $ must satisfy Jacobi's identity (other conditions for Lie brackets are clearly satisfied).
Let us study how this condition applies by examining the 3-dimensional Lie algebras, for which we have a complete list as reviewed in Sec.\,\ref{section:3D}.

\subsection{Three-dimensional systems: deformation of $\mathfrak{so}(3)$}
\label{subsec:3D_deformation}

We show that all types of the 3-dimensional Lie algebras can be derived by the \emph{deformations} from one \emph{simple Lie algebra}.
The ``mother'' is the type-IX algebra
(denoted by $\mathfrak{g}_{\mathrm{IX}}$, which is nothing but $\mathfrak{so}(3)$) that is characterized by
\begin{equation}
[\bm{e}_i, \bm{e}_j]_{\mathrm{IX}} = \epsilon_{ijk}\bm{e}_k.
\label{Lie-algebra-IX}
\end{equation}
In vector-analysis notation, we may write $[\bm{a}, \bm{b}]_{\mathrm{IX}} = \bm{a}\times\bm{b}$.

The multiplication table of the deformed bracket is given by calculating
\begin{equation}
[\bm{e}_i, \bm{e}_j]_{M} = M^{\mathrm{T}} [\bm{e}_i, \bm{e}_j]_{\mathrm{IX}} .
\label{bracket_of_g_M}
\end{equation}
If this bracket satisfies Jacobi's identity, we obtain a deformed Lie algebra, which we will denote by $\mathfrak{g}_M$.
Before discussing the Jacobi constraint need for the possible deformation matrix $M^{\mathrm{T}}$, 
we derive it directly from the multiplication tables of the Lie algebras (see Table\,\ref{table:Bianch-3D_by_deformation}).
The following relations are readily deduced:
\begin{enumerate}
\item
For a Class-A algebra, $M$ is symmetric,
while for a Class-B algebra, $M$ is non-symmetric.
This fact brings about fundamental differences between both classes; to be discussed later.

\item
Let $\mathfrak{g}'$ denote the derived algebra of $\mathfrak{g}$
(which is the ideal of $\mathfrak{g}$ consisting of elements such that $[\bm{e}_j,\bm{e}_k]$).
Since $\mathfrak{g}_{\mathrm{IX}}' = \mathfrak{g}_{\mathrm{IX}}$,
\[
\mathrm{dim}\, \mathfrak{g}_M' = \mathrm{Rank}\, M ,
\]
where $\mathrm{Rank}\, M= 3-\mathrm{dim}\,\mathrm{Ker}\,M = 3-\mathrm{dim}\,\mathrm{Coker}\,M^{\mathrm{T}} $.
In the table, 
\begin{eqnarray*}
& & \mathrm{dim}\,\mathfrak{g}_{\mathrm{I}}' = 0, \\
& & \mathrm{dim}\,\mathfrak{g}_{\mathrm{II}}' = \mathrm{dim}\,\mathfrak{g}_{\mathrm{III}}' = 1,\\
& & \mathrm{dim}\,\mathfrak{g}_{\mathrm{III}}' = \cdots =\mathrm{dim}\,\mathfrak{g}_{\mathrm{VII}}' = 2,\\
& & \mathrm{dim}\,\mathfrak{g}_{\mathrm{VIII}}' = \mathrm{dim}\,\mathfrak{g}_{\mathrm{IX}}' = 3.\\
\end{eqnarray*}
\end{enumerate}

\begin{table}
\caption{Bianchi classification of 3-dimensional Lie algebra.
Here we unify the classification by a general matrix $M$.
}
\vspace*{5mm}
\includegraphics[scale=0.6]{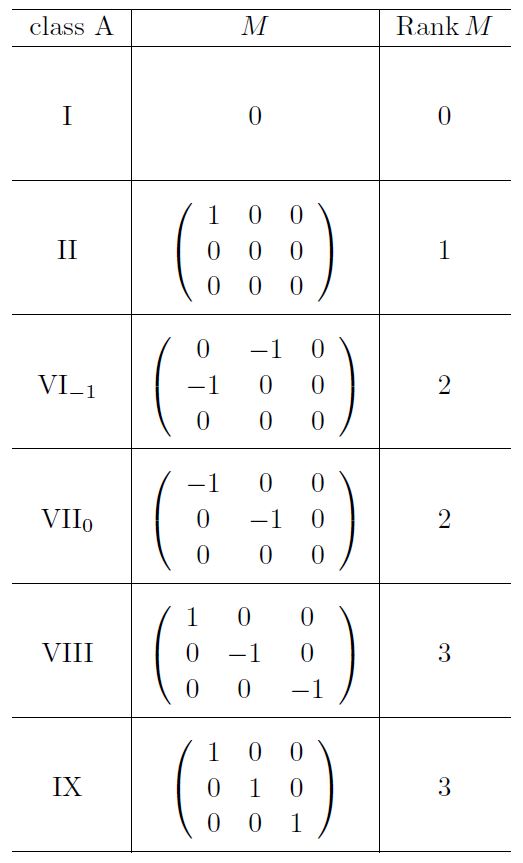}~~~
\mbox{\raisebox{11mm}{\includegraphics[scale=0.6]{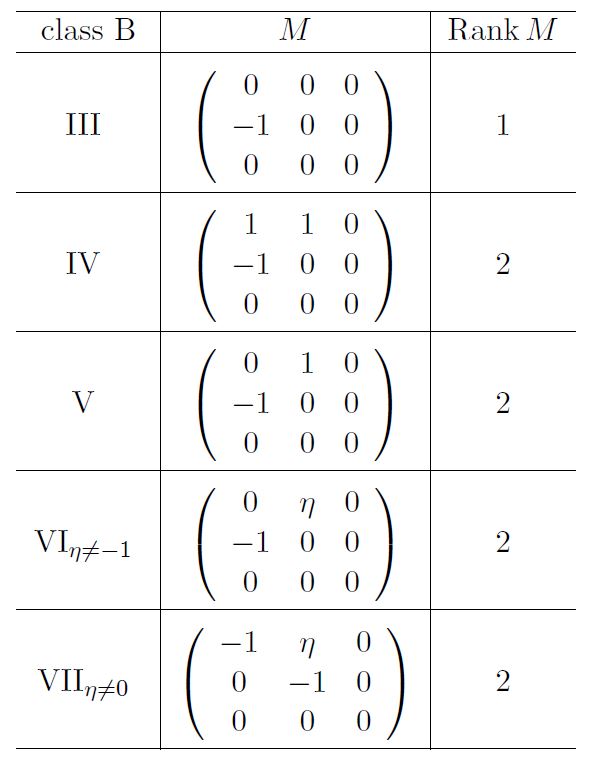}}}
\label{table:Bianch-3D_by_deformation}
\end{table}

Now we examine the conditions on $M$ for $[~,~]_M =M^{\mathrm{T}}[~,~]_{\mathrm{IX}}$ to be a Lie bracket. 
We can produce all 3-dimensional Lie algebras by the following process.
The dimension of the derived algebra plays the role of a guide
(cf.\,\cite{Jacobson}).
\begin{enumerate}
\item
For $\mathrm{dim}\,\mathfrak{g}_{M}' = 3$, only symmetric $M$ is allowed;
otherwise, Jacobi's identity does not hold.
Since $M^{\mathrm{T}}$ must not have a kernel (to obtain $\mathrm{dim}\,\mathfrak{g}_{M}' = 3$),
we have to demand
\[
[[\bm{e}_1,\bm{e}_2]_M , \bm{e}_3]_{\mathrm{IX}} +
[[\bm{e}_2,\bm{e}_3]_M , \bm{e}_1]_{\mathrm{IX}}  +
[[\bm{e}_3,\bm{e}_1]_M , \bm{e}_2]_{\mathrm{IX}}  = 0.
\]
Inserting $[\bm{e}_j,\bm{e}_k]_M = M^{\mathrm{T}} [\bm{e}_j,\bm{e}_k]_{\mathrm{IX}} $, the left-had side reads
\[
(M^{\mathrm{T}}_{23}-M^{\mathrm{T}}_{32}) \bm{e}_1 + 
(M^{\mathrm{T}}_{31}-M^{\mathrm{T}}_{13}) \bm{e}_2 +
(M^{\mathrm{T}}_{12}-M^{\mathrm{T}}_{21}) \bm{e}_3 .
\]
Hence we need $M^{\mathrm{T}}_{jk}=M^{\mathrm{T}}_{kj}$ for all $j\neq k$.
From this observation, it is also evident that, for every symmetric $M$, regardless of its rank,
$[~,~]_M =  M^{\mathrm{T}}[~,~]_{\mathrm{IX}}$ is a Lie bracket.
Hence, all class-A algebras are produced by some symmetric $M$.
For degenerate $M$ (i.e.\  for $\mathrm{dim}\,\mathfrak{g}_{M}' < 3$), however,
the symmetry condition can be weakened, and  some non-symmetric $M$ can still define Lie algebras.

\item
To define $\mathrm{dim}\,\mathfrak{g}_{M}' = 2$, 
we suppose the  matrix $M$  is  rank 2 so  that $\mathrm{Ker}\,M =\mathrm{Coker}\,M (= \mathrm{Ker}\,M^{\mathrm{T}}) = \{\bm{e}_3\}$,
i.e.\  $M=N\oplus0$ with a regular $2\times2$ matrix $N$
(notice that all rank-2 matrices $M$ of Table\,\ref{table:Bianch-3D_by_deformation} have such forms;
the reason why we need this setting will become clear in  the following construction).
Then, $\mathfrak{g}_{M}' $ is abelian;
for $\bm{e}_1, \bm{e}_2 \in \mathfrak{g}_{M}' $,
\begin{equation}
[\bm{e}_1,\bm{e}_2]_M = M^{\mathrm{T}} [\bm{e}_1,\bm{e}_2]_{\mathrm{IX}} = M^{\mathrm{T}} \bm{e}_3=0.
\label{2D_abelian_derived_algebra}
\end{equation}
The multiplication table is completed by evaluating
${[} \bm{e}_1, \bm{e}_3{]}_M$ and ${[} \bm{e}_2, \bm{e}_3{]}_M$.
By   definition,
\[
[ \circ, \bm{e}_3]_M = M^{\mathrm{T}} [\circ, \bm{e}_3]_{\mathrm{IX}}.
\]
Since the range of the operator on the right-hand side  excludes 
$\mathrm{Coker}\, M^{\mathrm{T}} = \mathrm{Ker}\, M=\{ \bm{e}_3\}$,
we see that $\mathfrak{g}_M'$ is indeed an ideal of $\mathfrak{g}_M$
(to put it in another way, if the condition $ \mathrm{Ker}\, M=\{ \bm{e}_3\}$ is violated,
it causes a contradiction with the  derived algebra $\mathfrak{g}_M'$ being  an ideal).
To evaluate $[ \circ, \bm{e}_3]_M$ for $\mathfrak{g}_M'$, we may define the 
$2\times2$ matrix
\[
A = N^{\mathrm{T}} [\circ, \bm{e}_3]_{\mathrm{IX}} = 
N^{\mathrm{T}} 
\left(\begin{array}{cc}
0 & -1 \\
~1~ & 0 \end{array}\right).
\]
This $A$ is identical to the matrix given in Jacobson\,\cite{Jacobson} Eq.\ (18), by which we obtain
\[
\left(\begin{array}{c}
{[} \bm{e}_1, \bm{e}_3{]}_M \\ {[}\bm{e}_2, \bm{e}_3{]}_M 
\end{array}\right) 
= A \left(\begin{array}{c}
\bm{e}_1 \\ \bm{e}_2
\end{array}\right) .
\]
Interestingly, for every regular matrix $A$ (thus, for every regular matrix $N$), 
the deformed product $[~,~]_M$ satisfies the Jacobi identity 
(each $[[\bm{e}_i,\bm{e}_j]_M,\bm{e}_k]_M$ vanishes separately),
so $\mathfrak{g}_M$ is a Lie algebra.
This is primarily due to the fact that the derived algebra $\mathfrak{g}_M'$ is abelian.
If the condition $ \mathrm{Ker}\, M^{\mathrm{T}} =\{ \bm{e}_3\}$ is violated,
(\ref{2D_abelian_derived_algebra}) does not hold, 
and then Jacobi's identity is not satisfied (see Remark\,\ref{remark:degenerate_M}).
Therefore, we do need both $ \mathrm{Ker}\, M^{\mathrm{T}} = \mathrm{Ker}\, M=\{ \bm{e}_3\}$
(i.e.\  $M=N\oplus0$ with regular $N$) to derive a $\mathrm{dim}\,\mathfrak{g}'_M=2$ algebra.

\item
For $\mathrm{dim}\,\mathfrak{g}_{M}' = 1$, 
a symmetric rank-1 $M$ defines a Class-A algebra, that is type-II (Heisenberg algebra).
There is  another possibility.
Let $\bm{e}_1$ be the element of the 1-dimensional $\mathfrak{g}_{M}'$.
Then, $\bm{e}_2, \bm{e}_3 \in \mathrm{Ker}\,M$.
Except for the symmetric one, the only possibility is of the form   
\[
M = 
\left(\begin{array}{ccc}
0 & 0 & 0 \\
-1 & ~0~ & ~0~ \\
0 & 0 & 0
\end{array}\right),
\]
which gives 
\[
[\bm{e}_1,\bm{e}_2]_M = 0, \quad
[\bm{e}_1,\bm{e}_3]_M = \bm{e}_1, \quad
[\bm{e}_2,\bm{e}_3]_M = 0.
\]
We may check that Jacobi's identity holds.  This is the type-III algebra.

\item
Evidently $M=0$ yields 
$\mathrm{dim}\,\mathfrak{g}_{M}' = 0$, which corresponds to the abelian type-I algebra.

\end{enumerate}

\begin{remark}[inadequate $M$]
\label{remark:degenerate_M}
\normalfont
For the derivation of $\mathrm{dim}\,\mathfrak{g}_M' =2$,
we assumed that $M$ is such that $N\oplus0$ with regular $N$.
Let us demonstrate that other types of degenerate $M$ deteriorate the deformation.

(1) First, consider
\[
M= \left(\begin{array}{ccc}
~1~ & ~1~ & ~0~ \\
0 & 0 & 1 \\
0 & 0 & 0
\end{array}\right).
\]
which has $\mathrm{Ker}\,M=\bm{e}_1-\bm{e}_2$, while $\mathrm{Ker}\,M^{\mathrm{T}}=\bm{e}_3$.
The multiplication table becomes
\[
[\bm{e}_1,\bm{e}_2]_M = 0, \quad
[\bm{e}_2,\bm{e}_3]_M = \bm{e}_1+\bm{e}_2, \quad
[\bm{e}_3,\bm{e}_1]_M = \bm{e}_3 \notin \mathfrak{g}'_M.
\]
Hence $\mathfrak{g}'_M$ fails to be an ideal.

(2) Take the transposed one as $M$:
\[
M= \left(\begin{array}{ccc}
~1~ & ~0~ & ~0~ \\
1 & 0 & 0\\
0 & 1 & 0
\end{array}\right).
\]
which has $\mathrm{Ker}\,M=\bm{e}_3$, while $\mathrm{Ker}\,M^{\mathrm{T}}=\bm{e}_1-\bm{e}_2$.
The multiplication table becomes
\[
[\bm{e}_1,\bm{e}_2]_M = \bm{e}_2, \quad
[\bm{e}_2,\bm{e}_3]_M = \bm{e}_1, \quad
[\bm{e}_3,\bm{e}_1]_M = \bm{e}_1 ,
\]
which violates the Jacobi's identity:
\[
[[\bm{e}_1,\bm{e}_2]_M , \bm{e}_3]_M +
[[\bm{e}_2,\bm{e}_3]_M , \bm{e}_1]_M +
[[\bm{e}_3,\bm{e}_1]_M , \bm{e}_2]_M = 
\bm{e}_1-\bm{e}_2.
\]
Notice that the residual is in $\mathrm{Ker}\,M^{\mathrm{T}}$.
\end{remark}

Summarizing the forgoing results, we have

\begin{theorem}[deformation of $\mathfrak{so}(3)$]
\label{theorem:3D}
Every 3-dimensional real Lie bracket can be written as $[~,~]_M=M^{\mathrm{T}}[~,~]_{\mathrm{IX}}$ with $M\in\mathrm{End}(\mathbb{R}^3)$ which is chosen from the following two classes:
\begin{enumerate}
\item
class A: $M$ is an arbitrary symmetric $3\times3$ matrix.
\item
class B: $M= N \oplus 0$ ($N$ is an arbitrary asymmetric $2\times2$ matrix).
\end{enumerate}
Accordingly, we have a unified representation of all 3-dimensional Lie-Poisson brackets:
\begin{equation}
\{ G,H \}_M = \langle [\partial_{\bm{\xi}} G, \partial_{\bm{\xi}} H]_M, \bm{\xi}\rangle 
=  \langle [\partial_{\bm{\xi}} G, \partial_{\bm{\xi}} H]_{\mathrm{IX}}, M\bm{\xi}\rangle .
\label{3D_Lie-Poisson_unified}
\end{equation}
The corresponding Poisson operator is
\begin{equation}
J_M(\bm{\xi}) \circ = J_{\mathrm{IX}}(M\bm{\xi}) \circ =[ \circ, M\bm{\xi}]^*_{\mathrm{IX}} = -(M\bm{\xi})\times \circ .
\label{3D_Poisson_op_unified}
\end{equation}
The singularity (where the rank of the Poisson operator becomes zero) is
\[
\sigma = \mathrm{Ker}\,M.
\]
\end{theorem}


\begin{corollary}[Casimirs of class-A Lie-Poisson brackets]
\label{cor:q-Casimir}
Let $M\in\mathrm{End}(\mathbb{R}^3)$ be a symmetric matrix (of any rank).  Then 
the Lie-Poisson bracket $\{G,H\}_M = \langle [\partial_{\bm{\xi}} G,\partial_{\bm{\xi}} H]_{\mathrm{IX}},M\bm{\xi}\rangle$
has a Casimir given by a quadratic form
\begin{equation}
C(\bm{\xi})= \frac{1}{2} \langle \bm{\xi},M \bm{\xi}\rangle .
\label{q-Casimir}
\end{equation}
\end{corollary}

\begin{proof}
By the symmetry of $M$, $\partial_{\bm{\xi}} C = M\bm{\xi}$.  Inserting this, we obtain
\[
J_M(\bm{\xi}) \partial_{\bm{\xi}} C = J_{\mathrm{IX}}(M\bm{\xi}) M \bm{\xi}  = -(M\bm{\xi})\times( M \bm{\xi})=0.
\]
\begin{flushright}
~~\qed
\end{flushright}
\end{proof}

\bigskip

This corollary does not preclude the existence of Casimirs for the class-B algebras;
as shown in Table\,\ref{table:Bianchi-Casimirs-B}, they are \emph{singular} functions
in the sense that each Casimir leaf contains the singularity $\sigma$ where $\mathrm{Rank}\,J(\bm{\xi})$ drops to zero.
As we will see in the next subsection, this singularity is related to the chirality of the spectra.

The 3-dimensional Lie algebras are special in that all of them have a unique mother $\mathfrak{so}(3)$,
and the symmetry/asymmetry of the deformation matrix $M$ determines the classification into A and B.
As we will see later (Sec.\,\ref{sec:dimension>3}), this is no longer true in higher dimensions, 
so that we will need to introduce ``class C'' to separate the classes A and B. 
Before extending to higher dimensions, we show how the class A algebras
yield Hamiltonian symmetric spectra around the singularities.
This property will be used as the determinant of class A in higher dimensions. 

\subsection{Spectra of 3-dimensional Lie-Poisson systems}
\label{subsec:linearization_around_sing-3D}
Let us analyze the spectra of the 3-dimensional Lie-Poisson systems linearized around the singularities
$\bm{\xi}_s \in \sigma = \mathrm{Ker}\,M$.
The aim is to prove the Hamiltonian symmetry for the class-A systems and, conversely, that this symmetry is generally broken for   class-B systems. As is easily inferred, the Hamiltonian symmetry of the class-A systems is due to the symmetry of the deformation matrix $M$.


Let  $H(\bm{\xi})$ be an arbitrary Hamiltonian (energy-Casimir functional), and denote
$\bm{h} =\partial_{\bm{\xi}} H|_{\bm{\xi}_s} $, which is a fixed vector.
The linearized equation (\ref{Hamilton-LS}) reads
\begin{equation}
\frac{\rmd}{\rmd t} \tilde{\bm{\xi}} =  [ \bm{h}, \tilde{\bm{\xi}} ]_M^* = [\bm{h}, M \tilde{\bm{\xi}}]_{\mathrm{IX}}^*.
\label{Hamilton-LS-M}
\end{equation}
Because type-IX is semi-simple, we may formally calculate as 
$[\bm{a},\bm{b}]_{\mathrm{IX}}^*
= [\bm{a},\bm{b}]_{\mathrm{IX}}=\bm{a}\times\bm{b}$
(see Remark\,\ref{remark:semi-simple}).
Therefore, the right-hand side of (\ref{Hamilton-LS-M}) reads
\begin{equation}
[ \bm{h}, \tilde{\bm{\xi}} ]_M^* =
[ \bm{h}, M \tilde{\bm{\xi}} ]_{\mathrm{IX}}^*
=
- [ M \tilde{\bm{\xi}},\bm{h} ]_{\mathrm{IX}}^*
=
- J_{\mathrm{IX}}(\bm{h}) M \tilde{\bm{\xi}} .
\label{LS-M-J_h}
\end{equation}
Notice that $ M \tilde{\bm{\xi}}$ is now regarded as a member of $X$.
With the constant-coefficient matrix $\mathscr{J}_{\bm{h}} = -J_{\mathrm{IX}}(\bm{h})$,
the linearized equation (\ref{Hamilton-LS-M}) can be written as
\begin{equation}
\frac{\rmd}{\rmd t} \tilde{\bm{\xi}} =  \mathscr{J}_{\bm{h}}  M \tilde{\bm{\xi}}.
\label{Hamilton-LS-M-2}
\end{equation}
By the definition, $\mathscr{J}_{\bm{h}}$ defines a homogeneous Poisson bracket
$\{ G, H \}_{\bm{h}} = \langle \partial_{\tilde{\bm{\xi}}} G,  \mathscr{J}_{\bm{h}} \partial_{\tilde{\bm{\xi}}} H \rangle$
(see Remark\,\ref{remark:generalized_Lie-Poisson}).
If $M$ is a symmetric matrix (class A), we can define a `Hamiltonian''
\[
\mathcal{C}(\tilde{\bm{\xi}}) = \frac{1}{2} \langle M \tilde{\bm{\xi}}, \tilde{\bm{\xi}} \rangle,
\]
by which the linearized equation (\ref{Hamilton-LS-M-2}) reads Hamilton's equation
\begin{equation}
\frac{\rmd}{\rmd t} \tilde{\bm{\xi}} =  \mathscr{J}_{\bm{h}}  \partial_{\tilde{\bm{\xi}}} \mathcal{C} (\tilde{\bm{\xi}}) .
\label{Hamilton-LS-M-Hamiltonian}
\end{equation}
Hence, the spectra of class-A have the Hamiltonian symmetry
(Remark\,\ref{remark:Hamiltonian_symmetry}).
The Hamiltonian $\mathcal{C}(\tilde{\bm{\xi}})$
is nothing but the Casimir evaluated for the perturbation (see Corollary\,\ref{cor:q-Casimir}).

Remembering Proposition\,\ref{proposition:conservation}, one may postulate that the other invariant,
the linearized energy $H_1(\tilde{\bm{\xi}}) =  \langle \bm{h} ,\tilde{\bm{\xi}} \rangle$
is the Casimir of $\mathscr{J}_{\bm{h}}$.
One can easily confirm that this is true.
It is remarkable that the roles of the Casimir and Hamiltonian are switched when   linearization around the singularity.

For the class-B Lie-Poisson systems, $M$ is not symmetric, so
$\mathscr{J}_{\bm{h}}  M$ is not a Hamiltonian generator;
hence, its spectrum need'nt have  the Hamiltonian symmetry. 
However,  $C(\tilde{\bm{\xi}})$ and $H_1(\tilde{\bm{\xi}}) $ are still invariant (Proposition\,\ref{proposition:conservation}).

In summary, we have the following corollary of Theorem\,\ref{theorem:3D}:

\begin{corollary}[Hamiltonian spectral symmetry]
\label{cor:Hamiltonian-spectrum}
A three-dimensional class-A Lie-Poisson system,
given by a symmetric deformation matrix $M$, has Hamiltonian symmetric spectra when linearized around a singular equilibrium point $\bm{\xi}_s \in \mathrm{Ker}(M)$.

\end{corollary}

\begin{remark}[linearized class-A system]
\label{remark:Hamiltonian-generator}
\normalfont
Corollary\,\ref{cor:Hamiltonian-spectrum} explains the observation in Sec.\,\ref{subsec:3D_spectra}.
The mathematical structure underlying the class-A linearized systems has the following two common ingredients 
that produce Hamiltonian symmetric spectra around the singularities:
\begin{enumerate}
\item
The full antisymmetry of the ``mother'' bracket $[\bm{x},\bm{\phi}]_{\mathrm{\rn{9}}} =  [\bm{x},\bm{\phi}]^*_{\mathrm{\rn{9}}} = \epsilon_{ijk}x^i\phi_j\bm{e}^k$, 
which is used in (\ref{LS-M-J_h}) to obtain the Poisson matrix 
$\mathscr{J}_{\bm{h}} = -J_{\mathrm{\rn{9}}}(\bm{h})$.
\item
The symmetry of the deformation matrix $M$, 
which is used in (\ref{Hamilton-LS-M-2}) to define the ``Hamiltonian'' 
$\frac{1}{2}\langle M\tilde{\bm{\xi}},\tilde{\bm{\xi}}\rangle$.
\end{enumerate}
\end{remark}

\section{Extension to higher dimensions}
\label{sec:dimension>3}

For dimension greater than three, the range of \emph{deformation} falls short of encompassing all possible Lie algebras.  
Yet, we can produce a class of Lie algebras (and the associated Lie-Poisson brackets)
by symmetric and asymmetric deformations of some fully antisymmetric Lie algebras.
We propose an extended classification:

\begin{definition}[classification into A, B and C]
\label{def:classification-ABC}
Let $\mathfrak{g}$ be an $n$-dimensional real Lie algebra.
\begin{itemize}
\item
If $\mathfrak{g}$ is fully antisymmetric (i.e.\  the Lie bracket is given by fully antisymmetric structure constants),
or it is the deformation of some fully antisymmetric Lie algebra by a symmetric matrix $M\in\mathrm{End}(\mathbb{R}^n)$,
we say that $\mathfrak{g}$ is class A.
\item
If $\mathfrak{g}$ is the deformation of some fully antisymmetric Lie algebra by an asymmetric matrix $M\in\mathrm{End}(\mathbb{R}^n)$,
we say that $\mathfrak{g}$ is class B.
\item
If $\mathfrak{g}$ is neither class A nor class B, 
we say that $\mathfrak{g}$ is class C.
\end{itemize}
\end{definition}

Remember that every 3-dimensional Lie algebra is either class A or class B,
because all of them are produced by deformations of the unique ``mother'' $\mathfrak{so}(3)$.
We can easily generalize Corollary\,\ref{cor:Hamiltonian-spectrum} to  arbitrary dimension:

\begin{theorem}[Hamiltonian spectral symmetry]
\label{theorem:Hamiltonian-spectrum}
Suppose that $\mathfrak{g}_M$ is a real $n$-dimensional class-A Lie algebra endowed with a
Lie bracket $[~,~]_M=M^{\mathrm{T}}[~,~]_{\mathrm{AS}}$, where $[~,~]_{\mathrm{AS}}$ is a fully antisymmetric Lie bracket,
and $M \in \mathrm{End}(\mathbb{R}^n)$ is a symmetric matrix.
Then, the linearized generator 
\[
\mathcal{A} = -[ \bm{h}, M \circ]^*_{\mathrm{AS}}
\quad (\bm{h}\in \mathfrak{g}_{M})
\]
has Hamiltonian symmetric spectra.
On the other hand, the linearization of class-B or class-C system has chiral (non-Hamiltonian) spectra.
\end{theorem}

The proof is evident from Remark\,\ref{remark:Hamiltonian-generator}.
Note that this theorem does not preclude the possibility of Hamiltonian symmetry of spectra in class-B or class-C systems;
special selection of $\bm{h}$ can produce symmetric spectra (see Table\,\ref{table:Bianchi-linear-B}).

By Remark\,\ref{remark:semi-simple}, we find 
\begin{corollary}[semi-simple Lie-Poisson system]
\label{col:semi-simple}
When a Lie-Poisson bracket $\{G,H\}=\langle [\partial_{\bm{\xi}}G, \partial_{\bm{\xi}} H]_{\mathfrak{g}}, \bm{\xi} \rangle$
is defined by the Lie bracket $[~,~]_{\mathfrak{g}}$ of a semi-simple Lie algebra $\mathfrak{g}$, 
it is Class-A, so that the linearized generator $\mathcal{A} = -[ \bm{h},  \circ]_{\mathfrak{g}}^* = -[ \bm{h},  \circ]_{\mathfrak{g}}$
($\bm{h}\in \mathfrak{g}$) has Hamiltonian symmetric spectra.
The Casimir $\frac{1}{2}|\bm{\xi}|^2$ is the Hamiltonian of the linearized system.
\end{corollary}

We also note that, unlike the case of 3-dimensional Lie algebras,
the deformation matrix $M$ is not so easily found as in Theorem\,\ref{theorem:3D}.
Even a symmetric $M$ may deteriorate Jacobi's identity.
Or, a small-rank $M$ such as $N\oplus^{n-2} 0$ no longer yields an abelian derived algebra, 
so the multiplication table of the deformed algebra must be carefully constructed to satisfy Jacobi's identity.

To see how the extended classification applies,
let us examine the 4-dimensional Lie algebras;
we invoke the complete list given in \cite{Patera}.

As is well known, there are no simple (or semi-simple) 4-dimensional Lie algebras 
(we exclude algebras that are direct sums of lower dimensional algebras).
For the ``mother'' algebra, we choose a fully antisymmetric algebra $\mathbb{R}\oplus\mathfrak{so}(3)$
(which is not in the list of \cite{Patera} because it has  the three-dimensional sub-algebra $\mathfrak{so}(3)$); 
the multiplication table of this algebra is
\[
\begin{array}{c|cccc}
~             & [\circ, \bm{e}_1] & [\circ, \bm{e}_2] & [\circ, \bm{e}_3] & [\circ, \bm{e}_4]  \\
\hline
\bm{e}_1 & 0                           &  0                         &  0                           & 0                           \\
\bm{e}_2 &  -                           &  0                         &  \bm{e}_4              & -\bm{e}_3              \\
\bm{e}_3 &  -                           &  -                         &  0                            &  \bm{e}_2              
\end{array}
\]
The linearized generator $[\bm{h}, \circ]^*= -[\circ, \bm{h}]^*$ (here we have used the full asymmetry)
has Hamiltonian spectra determined by the characteristic equation $\lambda^2 (\lambda^2 + |\bm{h}|^2)=0$.
All possible deformation matrices and the resultant multiplication tables are listed below:
\begin{enumerate}
\item
Symmetric deformation yielding $A_{4,10}$ (class A) algebra: 
\[
M^{\mathrm{T}} = \left( \begin{array}{cccc}
~0~ & ~0~ & ~0~ & ~1~ \\
0 & 1 & 0 & 0 \\
0 & 0 & 1 & 0 \\
1 & 0 & 0 & 0
\end{array} \right) , 
\quad\quad
\begin{array}{c|cccc}
~             & [\circ, \bm{e}_1] & [\circ, \bm{e}_2] & [\circ, \bm{e}_3] & [\circ, \bm{e}_4]  \\
\hline
\bm{e}_1 & 0                           &  0                         &  0                           & 0                           \\
\bm{e}_2 &  -                           &  0                         &  \bm{e}_1              & -\bm{e}_3              \\
\bm{e}_3 &  -                           &  -                         &  0                            &  \bm{e}_2              
\end{array}
\]
The linearized generator 
has Hamiltonian spectra determined by the characteristic equation $\lambda^2 (\lambda^2 + (h^4)^2)=0$.

\item
Symmetric deformation yielding $A_{4,8}$ (class A) algebra: 
\[
M^{\mathrm{T}} = \left( \begin{array}{cccc}
~0~ & ~0~ & ~0~ & ~1~ \\
0 & 0 & -1 & 0 \\
0 & -1& 0 & 0 \\
1 & 0 & 0 & 0
\end{array} \right) , 
\quad\quad
\begin{array}{c|cccc}
~             & [\circ, \bm{e}_1] & [\circ, \bm{e}_2] & [\circ, \bm{e}_3] & [\circ, \bm{e}_4]  \\
\hline
\bm{e}_1 & 0                           &  0                         &  0                           & 0                           \\
\bm{e}_2 &  -                           &  0                         &  \bm{e}_1              & \bm{e}_2               \\
\bm{e}_3 &  -                           &  -                         &  0                            & -\bm{e}_3              
\end{array}
\]
The linearized generator 
has Hamiltonian spectra determined by the characteristic equation $\lambda^2 (\lambda-h^4)(\lambda+h^4)=0$.

\item
Symmetric deformation yielding $A_{4,1}$ (class A) algebra: 
\[
M^{\mathrm{T}} = \left( \begin{array}{cccc}
~0~ & ~0~ & ~0~ & ~1~ \\
0 & -1 & 0 & 0 \\
0 & 0 & 0 & 0 \\
1 & 0 & 0 & 0
\end{array} \right) ,
\quad\quad
\begin{array}{c|cccc}
~             & [\circ, \bm{e}_1] & [\circ, \bm{e}_2] & [\circ, \bm{e}_3] & [\circ, \bm{e}_4]  \\
\hline
\bm{e}_1 & 0                           &  0                         &  0                           & 0                           \\
\bm{e}_2 &  -                           &  0                         &  \bm{e}_1              & 0                           \\
\bm{e}_3 &  -                           &  -                         &  0                            &  -\bm{e}_2              
\end{array}
\]
The linearized generator has only zero eigenvalue.

\item
Asymmetric deformation yielding $A_{4,3}$ (class B) algebra: 
\[
M^{\mathrm{T}} = \left( \begin{array}{cccc}
~0~ & -1& ~0~ & ~0~ \\
0 & 0 & 0 & 1 \\
0 & 0 & 0 & 0 \\
0 & 0 & 0 & 0
\end{array} \right) , 
\quad\quad
\begin{array}{c|cccc}
~             & [\circ, \bm{e}_1] & [\circ, \bm{e}_2] & [\circ, \bm{e}_3] & [\circ, \bm{e}_4]  \\
\hline
\bm{e}_1 & 0                           &  0                         &  0                           & 0                           \\
\bm{e}_2 &  -                           &  0                         &  \bm{e}_2              & 0                           \\
\bm{e}_3 &  -                           &  -                         &  0                            &  -\bm{e}_1              
\end{array}
\]
The linearized generator 
has chiral spectra determined by the characteristic equation $\lambda^3 (\lambda - h^3)=0$.
\end{enumerate}

As shown in Table I of \cite{Patera}, there are twelve 4-dimensional real Lie algebras
(excluding those which are direct sums of lower-dimensional algebras).
Separating out the aforementioned four algebras, the remaining eight are class C,
i.e.\  they are not obtained by any deformation of a fully antisymmetric 4-dimensional Lie algebra.
As easily inferred, the linearized generator is not Hamiltonian.
For example, $A_{4,12}$ algebra: 
\[
\begin{array}{c|cccc}
~             & [\circ, \bm{e}_1] & [\circ, \bm{e}_2] & [\circ, \bm{e}_3] & [\circ, \bm{e}_4]  \\
\hline
\bm{e}_1 & 0                           &  0                         &  \bm{e}_1              & -\bm{e}_2              \\
\bm{e}_2 &  -                           &  0                         &  \bm{e}_2              & \bm{e}_1               \\
\bm{e}_3 &  -                           &  -                         &  0                            &  0              
\end{array}
\]
is class C.  The characteristic equation of the linearized generator
is $\lambda^2 [(\lambda - h^3)^2 + (h^4 \lambda)^2]=0$, which gives a chiral spectrum.


\section{Vector bundle of $\mathfrak{so}(3)$ fibers and its deformations}
\label{sec:so(3)_bundle}

Here we give an example of an  infinite-dimensional Poisson manifold that is relevant to vortex dynamics in fluids.

\subsection{Vector bundle}
\label{subsubsec:so(3)-bundle}

We introduce a base space $\Omega \subset \mathbb{R}^3$, which is
a bounded  domain with a smooth boundary $\partial\Omega$.
We consider the vector bundle $E$ that consists of fibers of the $\mathfrak{so}(3)$ algebra;
each fiber has the Lie bracket 
\[
[\bm{a}, \bm{b} ]_{\RN{9}} = \bm{a}\times \bm{b} 
\quad (\bm{a}, \bm{b} \in \mathbb{R}^3).
\]
We assume that each $\bm{v} \in E$ is a $C^\infty$-class 3-vector function of $\bm{x}\in \Omega$, 
and write it as $\bm{v}(\bm{x})$.
Then, $E$ is regarded as a function space (totality of $C^\infty$-class cross-sections) endowed with a Lie bracket
\[
\dbracket{\bm{v}(\bm{x})}{\bm{w}(\bm{x})}_{\RN{9}} =   \bm{v}(\bm{x})\times \bm{w}(\bm{x}) ,
\quad (\bm{x}\in\Omega).
\]

The $L^2$-completion of $E$ is denoted by $V$.
Taking the $L^2$ inner product as the paring $\langle~,~\rangle$, 
the phase space is $V^* = V$.
Evidently, $\dbracket{~}{~}_{\RN{9}}^* =  \dbracket{~}{~}_{\RN{9}}$.
For a functional $F\in C^\infty(V^*)$, we define the gradient $\partial_{\bm{u}} F \in V$ by
\[
\delta F = F(\bm{u}+\epsilon\tilde{\bm{u}}) - F(\bm{u}) =
\epsilon \langle \partial_{\bm{u}} F , \tilde{\bm{u}} \rangle + O(\epsilon^2)
\quad (\forall \tilde{\bm{u}} \in V^*).
\]

The ``mother'' Lie-Poisson bracket (which will be deformed in various ways) is
\begin{equation}
\dpoisson{F}{G}_{\RN{9}} := \langle \dbracket{\partial_{\bm{u}} F}{\partial_{\bm{u}}G}_{\RN{9}}, \bm{u} \rangle
= \langle \partial_{\bm{u}} F, \dbracket{\partial_{\bm{u}}G}{\bm{u}}_{\RN{9}}^* \rangle,
\label{SO(3)_of_fields}
\end{equation}
and the corresponding Poisson operator is
\begin{equation}
\mathcal{J}_{\RN{9}} (\bm{u}) = \dbracket{\,\circ \, }{\bm{u}}_{\RN{9}}^*
= (\,\circ\, \times \bm{u} ).
\label{Poisson_oprator_IX_for_fields}
\end{equation}
We may evaluate the brackets on the dense subset $E\subset V^*=V$.


\subsection{Local deformations}
\label{subsec:chiral}

By applying the deformation using  a $3\times3$ constant-coefficient matrix $M$ of a type specified in Theorem\,\ref{theorem:3D}, we obtain a bundle of 3-dimensional Lie algebras.
Each of them is just the ``direct sum'' of the local Lie algebras; hence Jacobi's identity evidently holds.

An asymmetric $M$ yields a bundle of class-B algebra, and the
corresponding Lie-Poisson system exhibits chirality.
Let us demonstrate this with type III.  
Using 
\[
M = \left( \begin{array}{ccc}
~0~ & ~0~ & ~0~ \\
-1 & 0 & 0 \\
0 & 0 & 0
\end{array} \right),
\]
(see Tabel\,\ref{table:Bianch-3D_by_deformation}), we define a bracket 
\begin{equation}
\dbracket{ \bm{v}(\bm{x})}{\bm{w}(\bm{x})}_{\RN{3}}
:= M^{\mathrm{T}} \dbracket{\bm{v}(\bm{x})}{\bm{w}(\bm{x})}_{\RN{9}} 
= M^{\mathrm{T}} (\bm{v}(\bm{x})\times  \bm{w}(\bm{x})) .
\label{III-fieldized}
\end{equation}
Evidently, this defines a Lie algebra on $E$.
The Lie-Poisson bracket (\ref{SO(3)_of_fields}) is deformed to
\begin{equation}
\dpoisson{F}{G}_{\RN{3}} = \langle \dbracket{\partial_{\bm{u}} F}{\partial_{\bm{u}} G}_{\RN{3}}, \bm{u} \rangle
= \langle \dbracket{\partial_{\bm{u}} F}{\partial_{\bm{u}} G}_{\RN{9}} , M \bm{u} \rangle,
\label{III_of_fields}
\end{equation}
which gives a Poisson operator
\begin{equation}
\mathcal{J}_{\RN{3}}(\bm{u})
= \,\circ\, \times (M\bm{u}) ,
\label{Poisson-operator_for_III_of_fields}
\end{equation}
where $M\bm{u}= (0~ -u_1~ 0)^{\mathrm{T}}$.
This deformed system exhibits chirality.  
The linearized equation around the singularity $u_1=0$ is
(denoting $\bm{h}= \partial_{\bm{u}} H |_{u_1=0}$)
\[
\frac{\partial}{\partial t} \left( \begin{array}{c}
\tilde{u_1} \\
\tilde{u_2} \\
\tilde{u_3} \end{array} \right)
= 
\left( \begin{array}{c}
-h^3 \tilde{u_1} \\
0 \\
h^1 \tilde{u_1} \end{array} \right),
\]
which generates a chiral solution $\tilde{u}_1 \propto \rme^{-h^3 t} $.
Here $h^3$ is a function of space $\bm{x}$, so $-h^3$ is a continuous spectrum.


\subsection{Deformation by the ``curl'' operator: vortex dynamics system}
\label{subsec:vortex}

Here we deform $\bm{u}\in E$ by a differential operator curl (to be denoted by $\nabla\times$),
and choose  \emph{vorticity} $\bm{\omega} = \nabla\times\bm{u}$ as our observable.
Then, we obtain the Lie-Poisson bracket of vortex dynamics.
We start by preparing the mathematical definition of the curl operator.

\subsubsection{Self-adjoint curl operator}
We consider a subspace $E_\Sigma \subset E$
consisting of smooth 3-vectors that are solenoidal ($\nabla\cdot\bm{v}=0$),
tangential to the boundary 
($\bm{n}\cdot\bm{v}=0$, where $\bm{n}\cdot$ is the trace of the normal component onto $\partial\Omega$),
and 0-flux ($\int_S \bm{\nu}\cdot\bm{v}\,\rmd^2 x=0$, 
where $S$ is an arbitrary cross-section of the handle, if any, of $\Omega$, and $\bm{\nu}\cdot$ is the 
trace of the normal component onto $S$).
Let $L^2_\Sigma (\Omega)$ be the Hilbert space given by the $L^2$-completion of $E_\Sigma$:
 \[
L^2_\Sigma (\Omega) = \{ \bm{v} \in L^2(\Omega); \, \nabla\cdot\bm{v}=0,\,\bm{n}\cdot\bm{v}=0,\,
\int_S \bm{\nu}\cdot\bm{v}\,\rmd^2 x=0 \}.
 \]
The orthogonal complement of $L^2_\Sigma(\Omega)$ is $\mathrm{Ker}\,(\nabla\times)$,
which we will denote by $L^2_\Pi(\Omega)$, i.e.
\[
L^2(\Omega) = L^2_\Sigma(\Omega) \oplus L^2_\Pi(\Omega).
\]
We denote by $\mathscr{P}_\Sigma$ the orthogonal projection onto $L^2_\Sigma (\Omega)$
(when operated, this projector subtracts the \emph{irrotational component}),
and $\mathscr{P}_\Pi = I-\mathscr{P}_\Sigma$.

To formulate a system of vortex dynamics,
we invoke the self-adjoint curl operator given by \cite{YG}.
Let 
\[
H^1_{\Sigma\Sigma}(\Omega) 
= \{ \bm{u}\in L^2_\Sigma(\Omega);\, \nabla\times\bm{u}\in L^2_\Sigma(\Omega) \} ,
\] 
which is a dense, relatively compact subset of $\in L^2_\Sigma(\Omega)$.
We define a self-adjoint operator in $L_\Sigma(\Omega)$ such that $\mathcal{S} \bm{u} = \nabla\times\bm{u}$
on the domain $H^1_{\Sigma\Sigma}(\Omega) $.
This is a surjection to $ L^2_\Sigma(\Omega)$ with a compact inverse $\mathcal{S}^{-1}$,
so the set of eigenfunctions of $\mathcal{S}$ gives an orthogonal complete basis of $L_\Sigma(\Omega)$.

Combining with $\mathscr{P}_\Sigma$, we consider $\mathcal{S} $ in $L^2(\Omega)$:
\begin{equation}
{\mathscr{S}} = \mathcal{S} \mathscr{P}_\Sigma.
\label{sa-curl_extended}
\end{equation}
We may write ${\mathscr{S}} = \mathcal{S} \oplus 0 \mathscr{P}_\Pi$.
Evidently, ${\mathscr{S}} $ is a self-adjoint operator in $V^*=L^2(\Omega)$
(notice that this ${\mathscr{S}} $ is different from the non-self-adjoint curl operator $T$ or $\tilde{T}$ discussed in\,\cite{YG}).

\subsubsection{Deformation by the self-adjoint curl operator}

Let us deform $\dpoisson{G}{H}_{\RN{9}} = \langle \dbracket{\partial_{\bm{u}} G}{\partial_{\bm{u}} H}_{\RN{9}}, \bm{u} \rangle$ to  $\langle \dbracket{\partial_{\bm{u}} G}{\partial_{\bm{u}}H }_{\RN{9}}, \mathscr{S} \bm{u} \rangle$
(which means that we deform the Poisson operator
$\mathcal{J}_{\RN{9}}(\bm{u})$ to
$\mathcal{J}_{\RN{9}}(\mathscr{S}\bm{u})$).
Consequently, the Lie bracket $\dbracket{\bm{u}}{\bm{v}}_{\RN{9}}$ of $E$ is deformed to 
\begin{equation}
\dbracket{\bm{\omega}}{\bm{\phi}}_{\mathscr{S}}
= \mathscr{S} \dbracket{\bm{\omega}}{\bm{\phi}}_{\RN{9}} 
\quad (\bm{\omega}, \bm{\phi} \in E_\Sigma).
\label{curl-deformed_so(3)-bundle}
\end{equation}
Notice that we define the Lie algebra on a reduced space $E_\Sigma=\mathscr{P}_\Sigma E$
(see Remark\,\ref{remark:reduction}).
On $E_\Sigma$, we may evaluate (using $\nabla\cdot \bm{\omega} = \nabla\cdot \bm{\phi}=0$)
\[
\mathscr{S} \dbracket{\bm{\omega}}{\bm{\phi}}_{\RN{9}} = \nabla\times (\bm{\omega} \times \bm{\phi})
= (\bm{\phi}\cdot\nabla) \bm{\omega} - (\bm{\omega}\cdot\nabla)\bm{\phi}.
\]
The right-hand side is nothing but the Lie derivative of the vector:
$\mathcal{L}_{\bm{\phi}} \bm{\omega} $.
Hence, Jacobi's identity is evident (being equivalent to the Leibniz law for Lie derivatives) 
(cf.\,\cite{Chandre2013}).


The reduction to $E_\Sigma$ is naturally implemented in the definition of the deformed Lie-Poisson bracket,
because $E_\Sigma$ can be regarded as the phase space of vorticities;
by the definition of the self-adjoint curl operator $\mathcal{S}$, we find
\[
E_\Sigma =\{ \bm{\omega} = \mathscr{S} \bm{u} ; \, \bm{u}\in E \}  .
\]
By the chain rule, we observe, for a functional $F(\bm{\omega}) \in C^\infty(E_\Sigma)$,
\[
\delta F = \langle\partial_{\bm{u}} F , \tilde{\bm{u}} \rangle
= \langle\partial_{\bm{\omega}} F , \tilde{\bm{\omega}} \rangle
= \langle\partial_{\bm{\omega}} F ,  \mathscr{S} \tilde{\bm{u}} \rangle
= \langle\mathscr{S} \partial_{\bm{\omega}} F ,  \tilde{\bm{u}} \rangle .
\]
Therefore, we may evaluate, for a functional $F(\bm{\omega}) \in C^\infty(E_\Sigma)$,
\[
\partial_{\bm{u}} F = \mathscr{S} \partial_{\bm{\omega}} F.
\]
We define the curl-deformed Lie-Poisson bracket on $ C^\infty(E_\Sigma)$:
\begin{eqnarray}
\dpoisson{G}{H}_{\mathscr{S}} &:=& 
\langle \dbracket{\partial_{\bm{u}} G}{\partial_{\bm{u}} H}_{\RN{9}}, \mathscr{S}\bm{u} \rangle
\nonumber \\
&=&
\langle \dbracket{\mathscr{S}\partial_{\bm{\omega}} G}{\mathscr{S}\partial_{\bm{\omega}} H}_{\RN{9}},  \bm{\omega} \rangle
\nonumber \\
&=&
\langle \mathscr{S} \partial_{\bm{\omega}} G, \dbracket{\mathscr{S}\partial_{\bm{\omega}} H}{\bm{\omega}}^*_{\RN{9}}\rangle
\nonumber \\
&=&
\langle \partial_{\bm{\omega}} G, 
\mathscr{S} \dbracket{\mathscr{S} \partial_{\bm{\omega}} H}{\bm{\omega}}^*_{\RN{9}} \rangle .
\label{SO(3)_of_fields-curled-curled}
\end{eqnarray}
The corresponding Poisson operator reads
\begin{equation}
\mathcal{J}_{\mathscr{S}}(\bm{\omega})
= \mathscr{S}((\mathscr{S}\,\circ\, )\times\bm{\omega} )
= \mathscr{P}_\Sigma \nabla\times( (\nabla\times\circ\,)\times\bm{\omega}),
\label{Poisson-operator_for_SO(3)_of_fields_curled-curled}
\end{equation}
which applies to the vortex dynamics equation for formulating a Hamiltonian form (see Remark\,\ref{remark:fluid_bracket}).
The Casimir is 
\begin{equation}
C(\bm{\omega}) = \frac{1}{2} \langle  {\mathcal{S}}^{-1} \bm{\omega}, \bm{\omega} \rangle,
\label{Casimir_for_SO(3)_of_fields_curled-curled}
\end{equation}
which is known as the \emph{helicity}, an important invariant of ideal (barotropic and dissipation-free) fluid motion. 

\begin{remark}[vortex dynamics]
\label{remark:fluid_bracket}
\normalfont
The ``Hamiltonian'' of incompressible fluid (mass density = 1) is given by
\begin{equation}
H(\bm{\omega}) = \frac{1}{2} \int_\Omega |\mathcal{S}^{-1} \bm{\omega}|^2\,\rmd^3x
 = \frac{1}{2} \int_\Omega |\bm{u}|^2\,\rmd^3 x.
\label{fluid_energy}
\end{equation}
Here, $\bm{u} =\mathcal{S}^{-1} \bm{\omega} \in L^2_\Sigma(\Omega)$ is the \emph{dynamical component}
of the fluid velocity (the irrotational component $\in L^2_\Pi(\Omega)$ is fixed by the boundary condition and the circulation law).
Hamilton's equation $\ddt {F} = \{ F, H \}_{\mathscr{S}}$
yields the vortex dynamics equation
\begin{equation}
\partial_t {\bm{\omega}} = -\nabla\times(\bm{\omega}\times\bm{u} ).
\label{3D_vortex_equation-2}
\end{equation}
The helicity
 \begin{equation}
C(\bm{\omega}) = \frac{1}{2} \int_\Omega \bm{\omega}\cdot \mathcal{S}^{-1} \bm{\omega}\,\rmd^3x
 = \frac{1}{2} \int_\Omega \mathscr{S}\bm{u}\cdot \bm{u}\,\rmd^3x 
\label{fluid_helicity}
\end{equation}
is a Casimir of the bracket $\dpoisson{G}{H}_{\mathscr{S}}$.
\end{remark}

Evidently, this curl-deformed system is class A.
The linearized equation reads, denoting $\bm{h} = (\mathscr{S} \partial_{\bm{\omega}} H)_{\bm{\omega}=0}$,
\begin{equation}
\partial_t {\tilde{\bm{\omega}}} = \mathscr{S}(\bm{h}\times \tilde{\bm{\omega}}) .
\label{3D_vortex_equation-linear}
\end{equation}
By the symmetry $\mathscr{S}(\bm{h}\times \tilde{\bm{\omega}}) = - \mathscr{S}(\tilde{\bm{\omega}}\times\bm{h})$,
we may rewrite the right-hand side of (\ref{3D_vortex_equation-linear}) as,
using the Poisson operator of (\ref{Poisson-operator_for_SO(3)_of_fields_curled-curled}) and
the Casimir of (\ref{Casimir_for_SO(3)_of_fields_curled-curled}),
\begin{equation}
\partial_t {\tilde{\bm{\omega}}} = -\mathcal{J}_{\mathscr{S}}(\bm{h}) \mathcal{S}^{-1} \tilde{\bm{\omega}}
=  -\mathcal{J}_{\mathscr{S}}(\bm{h}) \partial_{ \tilde{\bm{\omega}}} C(\tilde{\bm{\omega}}),
\label{3D_vortex_equation-linear-2}
\end{equation}
which is a linear Hamiltonian system with the Casimir as the Hamiltonian
(the linear operator has a continuous spectrum due to flow shear; cf.\,\cite{BalmforthMorrison,Morrison2003}).


\begin{remark}[reduction]
\label{remark:reduction}
\normalfont
In the definition (\ref{curl-deformed_so(3)-bundle}) of the curl-deformed Lie bracket
$\dbracket{~}{~}_{\mathscr{S}}$,
we \emph{reduced} the state space from $E$ to $E_\Sigma=\mathscr{P}_\Sigma E=\mathscr{S}E$.
In some sense, this means that we are considering a derived algebra,
or, the ideal consisting of members such that $\mathscr{S}\dbracket{\bm{v}}{\bm{w}}_{\RN{9}}$
($\forall \bm{u}, \bm{v} \in E$,
while $\mathscr{S}\dbracket{\bm{v}}{\bm{w}}_{\RN{9}}$ is not a Lie bracket on $E$).
If we apply a similar reduction to the 3-dimensional Lie algebras  discussed in Sec.\,\ref{subsec:3D_deformation},
i.e.\ if we evaluate the deformed bracket $[~,~]_M = M^{\mathrm{T}}[~,~]_{\RN{9}}$ on $M^{\mathrm{T}} X$,
it becomes abelian (as mentioned when we derived rank-2 systems).
However, the present example of reduction yields the non-abelian algebra.
\end{remark}

\subsection{Variety of vortex systems}
\label{subsec:combined-symmetric}
We may modify the standard curl operator to a variety of differential operators.
by which we can formulate generalized vortex systems. 
We consider a symmetric deformation by a combined self-adjoint operator
\begin{equation}
\mathscr{M} = M^{\mathrm{T}} \mathscr{S} + \mathscr{S} M ,
\label{combination}
\end{equation}
where $M$ is some deformation matrix (either symmetric or asymmetric; see Table.\,\ref{table:Bianch-3D_by_deformation}).

We define a deformed vorticity
\begin{equation}
\bm{\omega}_{M} = \mathscr{M} \bm{u} ,
\quad (\bm{u} \in E).
\label{deformed_vorticity}
\end{equation}
The totality of the deformed vorticity constitute a phase space:
\begin{equation}
E_{\mathscr{M}} = \{ \bm{\omega}_M=\mathscr{M} \bm{u};\, \bm{u}\in E \}.
\label{space_of_deformed_vorticity}
\end{equation}

On $C^\infty(E_{\mathscr{M}})$, we define a deformed Lie-Poisson bracket
(denoting $\bm{\omega}_M=\mathscr{M}\bm{u}$)
\begin{eqnarray*}
\dpoisson{G}{H}_{\mathscr{M}} &:=& 
\langle \dbracket{\partial_{\bm{u}} G}{ \partial_{\bm{u}} H}_{\RN{9}} ,\mathscr{M} \bm{u}\rangle
\\
&=& \langle \dbracket{\mathscr{M}\partial_{\bm{\omega}_{M}}G}{ \mathscr{M}\partial_{\bm{\omega}_{M}} H}_{\RN{9}} , \bm{\omega}_{M} \rangle
\\ 
&=&\langle \partial_{\bm{\omega}_{M}} G, \mathscr{M}\dbracket{\mathscr{M}\partial_{\bm{\omega}_M} H}{\bm{\omega}_{M} }^*_{\RN{9}} \rangle .
\end{eqnarray*}
The corresponding Poisson operator is
\begin{equation}
\mathcal{J}_{ \mathscr{M}}(\bm{\omega}_M)
= \mathscr{M}((\mathscr{M}\,\circ\, )\times\bm{\omega}_M ),
\label{Poisson-operator_for_SO(3)_of_fields_curled-curled-M}
\end{equation}
which has a Casimir 
\[
C(\bm{u}) = \frac{1}{2} \langle  \mathscr{M}_0^{-1} \bm{\omega}_M, \bm{\omega}_M \rangle,
\]
where $\mathscr{M}_0=\mathscr{M}/\mathrm{Ker}(\mathscr{M})$, so that 
$\mathscr{M} \mathscr{M}_0^{-1} \bm{\omega}_M = \bm{\omega}_M$.
This symmetric Casimir plays the role of the Hamiltonian in the linearized system, resulting in the Hamiltonian symmetry of the spectra.

\section{Conclusion}
The Lie-Poisson algebra is a special class of Poisson algebras, which is naturally introduced to the phase space $X^*$ that is dual to some Lie algebra $X$.
The coadjoint action, generated by a Hamiltonian (a smooth function on $X^*$), describes the evolution of an observable (point $\bm{\xi} \in X^*$).
The problem we have explored is how the Lie algebra $X$ is \emph{deformed} when we transform the observable $\bm{\xi} \mapsto M\bm{\xi}$ by $M \in \mathrm{End}(X^*)$.
Guided by Bianchi's list of Lie algebras, we found that 
the symmetry/asymmetry of the deformation matrix $M$ gives an interesting classification A/B,
which corresponds to the Hamiltonian symmetry/asymmetry of spectra in the neighborhood of the singularity (nullity) of the coadjoint action;
the symmetry breaking, occurring in class-B systems, appears as chirality (breaking of the time-reversal symmetry) in the neighborhood of the singularity, which is forbidden in usual Hamiltonian spectra evaluated around regular equilibrium points (critical points of a given Hamiltonian).
Since the linearization works out differently in the neighborhood of singularities 
(which commonly exists in Lie-Poisson manifolds, but the dimension of the set of singularities depends on the Lie algebra), there is no general reason for the spectra to have the Hamiltonian symmetry.
Therefore, it is more interesting that class-A systems maintain the Hamiltonian symmetry
(Remark\,\ref{remark:Hamiltonian-generator} explains how this occurs).

The deformation induced by $M\in\mathrm{End}(X^*)$ is different from a ``coordinate change'' in a Lie algebra
(or the isomorphic deformations; see\, \cite{deformation});
in the latter, $T\in\mathrm{Aut}(X)$ applies as $[\bm{x}, \bm{y}] \mapsto [\bm{x}', \bm{y}']_T := T[T^{-1}\bm{x}', T^{-1} \bm{y}']$.
It is also different from a reduction (homomorphism to a sub-algebra; see \cite{Marsden}, as well as Remark\,\ref{remark:reduction})
or a constraint yielding a Dirac bracket (see \cite{Chandre2013}).
It is a back reaction to the Lie algebra caused by deforming the observable.  
Possible deformations $M$ are rather restricted (even for symmetric ones) by guaranteeing Jacobi's identity for the deformed bracket $M^{\mathrm{T}}[~,~]$. 
Interestingly, however, the 3-dimensional Lie algebras are totally derived from $\mathfrak{so}(3)$ by some deformations.
Given a base space, this mother algebra produces a variety of field models;
the deformation by the self-adjoint curl operator, for example, yields the Poisson manifold (infinite-dimensional) of vortex dynamics. However, it is challenging to obtain class-B infinite-dimensional systems and demonstrate chirality in such field theories. This will be discussed in future work.

\acknowledgments
The authors acknowledge stimulating discussions with Tadashi Tokieda and Richard Montgomery. 
In addition,  they would like to acknowledge insightful  comments of an anonymous referee and PJM would like to thank A. M. Bloch and Alan Weinstein for helpful correspondence. 
This material is based upon work supported by the National Science Foundation under Grant No.~1440140, while the authors were in residence at the Mathematical Sciences Research Institute in Berkeley, California, during the semester ``Hamiltonian systems, from topology to applications through analysis'' year of 2018.
The work of ZY was partly supported by JSPS KAKENHI grant number 17H01177,
and that of PJM was supported by the  DOE Office of Fusion Energy Sciences, under DE-FG02-04ER- 54742.



\end{document}